\documentclass[11pt]{amsart}  

\usepackage{euler}
\usepackage{amscd} 
\usepackage{mathrsfs}
\usepackage{amsmath,amsbsy,amsthm,amsfonts,amssymb,bm}
\usepackage{hyperref}
\usepackage{graphicx}
\newtheorem{theorem}{Theorem}[section]
\newtheorem{lemma}[theorem]{Lemma}
\newtheorem{proposition}[theorem]{Proposition}
 
\theoremstyle{definition}
\newtheorem{definition}[theorem]{Definition}
\newtheorem{example}[theorem]{Example}
\newtheorem{remark}[theorem]{Remark} 
\numberwithin{equation}{section}

\usepackage{color}

\def\HD{H_{D}}
\def\HDop{\HD^{\oplus}}
\def\UD{U^{D}}
\def\UN{U^{N}}
\def\OmQ{\Omega}
\def\SQ{S}
\def\OmC{\Omega_{cl}}
\def\SC{S_{cl}}
\def\PP{\mathbb{\Pi}}
\def\BB{\mathbb{B}}

\newcommand{\pd}[2]{\frac{\partial {#1}}{\partial {#2}}}

\def\unit{\bm{1}}
\def\fv{\mathbf{f}}

\def\MC{\mathbb{S}}

\def\Hfree{H_{0}}

\def\Ufree{U^{0}}

\def\HK{H_K}

\def\Gc{\breve{G}}

\def\GG{\mathcal{G}}

\def\eps{\varepsilon}

\def\z*{\bar z}

\def\S{\mathcal S}

\def\C{\mathcal C}

\def\dom{\text{\rm dom}}
\def\ran{\text{\rm ran}}

\def\RE{\mathbb R}
\def\CO{{\mathbb C}}

\def\dd{\displaystyle}

\def\ph*{\phi_\star}

\def\be{\begin{equation}}
\def\ee{\end{equation}}
\def\min{{\rm min}}

\def\-{{\rm in}}
\def\+{{\rm ex}}
\newcommand{\sgn}{\operatorname{sgn}}

\def\Im{\operatorname{Im}}
\def\Re{\operatorname{Re}}
\def\sigmap{\breve{\sigma}}
\def\sigmaq{\sigma}
\def\sigmazero{\sigma_0}
\def\Szero{A}

\def\RD{R^{D}}
\def\LD{L_D}

\DeclareMathOperator*{\slim}{s-lim}

\title[The semiclassical limit on a star-graph with Kirchhoff conditions]{The semiclassical limit on a star-graph with Kirchhoff conditions}
\author{Claudio Cacciapuoti}
\address{DiSAT, Sezione di Matematica, Universit\`a dell'Insubria, via Valleggio 11, I-22100
Como, Italy}
\email{claudio.cacciapuoti@uninsubria.it}

\author{Davide Fermi}
\address{Classe di Scienze, Scuola Normale Superiore, Piazza dei Cavalieri 7, I-56126 Pisa, Italy
}
\email{fermidavide@gmail.com}

\author{Andrea Posilicano}
\address{DiSAT, Sezione di Matematica, Universit\`a dell'Insubria, via Valleggio 11, I-22100
Como, Italy}
\email{andrea.posilicano@uninsubria.it}

\thanks{The authors acknowledge the support of the National Group of Mathematical Physics (GNFM-INdAM)}

\begin{document}

\begin{abstract} 
We consider the dynamics of a quantum particle of mass $m$ on a $n$-edges  star-graph with Hamiltonian  $H_K=-(2m)^{-1}\hbar^2 \Delta$ and Kirchhoff conditions in the vertex. We describe  the semiclassical limit of the quantum evolution of an initial state supported on one of the edges and close to a Gaussian coherent state. We define the limiting  classical dynamics through a Liouville operator on the graph, obtained by means of Kre\u{\i}n's theory of singular perturbations of self-adjoint operators. For the same class of initial states, we study the semiclassical limit of the wave and scattering operators for the couple $(H_K,\HDop)$, where $\HDop$ is the free Hamiltonian with Dirichlet conditions in the vertex.  
\end{abstract}

\maketitle

\begin{footnotesize}

\end{footnotesize}

\section{Introduction}
Aim of this work is to provide the semiclassical dynamics and scattering for an approximate coherent state  propagating freely on a star-graph, in presence of Kirchhoff conditions in the vertex.

Since the pioneering work of Kottos and Smilansky  \cite{KS-prl97}, having in mind applications to quantum chaos,  the semiclassical limit of quantum graphs is often understood as the study of the distribution of  eigenvalues (or resonances, see \cite{KS-prl00}) of self-adjoint realizations of $-(2m)^{-1}\hbar^2 \Delta$ on the graph.

To the best of our knowledge, a first study of the semiclassical limit for quantum dynamics on graphs is due to Barra and Gaspard \cite{BG-pre01_2} (see also \cite{BG-pre01_1}, where the limiting classical model is comprehensively discussed). In this case, the semiclassical limit  is understood in terms of the convergence of a  Wigner-like function for graphs when $\hbar$ (the reduced Planck constant) goes to zero.
 
 Inspired by the work of Hagedorn \cite{Hag1}, instead, we  look directly at the dynamics of the wave-function,  for a class of initial states which are close to Gaussian coherent states supported  on one of the edges of the graph.

Closely related to our work is a series of papers by  Chernyshev and Shafarevich \cite{C-psim10, CS-rjmp08, CS-ptrs14} in which the authors study the $\hbar\to0$ limit of Gaussian wave packets propagating on graphs. Their main interest is the asymptotic growth (for large times) of the number of wave packets propagating on the graph. The main tool for the analysis is the complex WKB method by Maslov (see \cite{Maslov}). We also point out the work \cite{CS-mz07}, by the same authors,  in which they study the small $\hbar$ asymptotics of the eigenvalues of Schr\"odinger operators on quantum graphs (with Kirchhoff conditions in the vertices and in the presence of potential terms).  

 In our previous work \cite{CFP19} we studied the semiclassical limit in the presence of a singular potential. Specifically, we considered the operator  $H_\alpha$,  which is the quantum Hamiltonian in $L^2(\RE)$ formally written as $H_\alpha = -\frac{\hbar^2}{2m}\,\Delta + \alpha \delta_0$, where $m$ is the mass of the particle,  $\delta_0$ is the Dirac-delta distribution centered in $x=0$, and $\alpha$ is a real constant measuring the strength of the potential. Given a Gaussian coherent state on the real line of the form 
\begin{equation}\label{initial}
\psi^{\hbar}_{\sigma,\xi}(x):=\frac{1}{(2\pi\hbar)^{1/4}\sqrt{\sigma}}\ \exp\left({-\frac{1}{4\hbar \sigmazero\sigma}\,(x-q)^{2}+\frac{i}{\hbar}\,p(x-q)}\right) \qquad \mbox{for\, $x\in\RE$}\,,
\end{equation}
with $\sigma\in\CO$, $\Re \sigma = \sigma_0 > 0$ and $\xi\equiv(q,p)\in\RE^2$,  we studied  the limit  $\hbar\to 0$  of $e^{-i\frac{t}{\hbar}H_{\alpha}}\,\psi^{\hbar}_{\sigmazero,\xi}$.

To this aim we reasoned as follows. For fixed $x \in \RE$, consider the classical wave function defined by
\begin{equation*}
\phi^{\hbar}_{\sigma,x}:\RE^{2}\to\CO\,,\quad
\phi^{\hbar}_{\sigma,x}(\xi):=\psi^{\hbar}_{\sigma,\xi}(x)\,. 
\end{equation*}
Consider the vector field $X_{0}(q,p)=(p/m,0)$ associated to the free classical Hamiltonian $h_{0}(q,p)= p^2/(2m)$ ($q$ is  the position  and $p$ the  momentum of the classical particle of mass $m$), and the Liouville operator
\begin{equation*}
\dom(L):= \C_0^\infty(\RE^2),\qquad L:=-\,i\,X_{0} \cdot \nabla =-\,i\,{p \over m}\,\frac{\partial  }{\partial q}\;.
\end{equation*}
We set 
\begin{equation*}
\xi_t:=\left(q+\frac{pt}{m},p\right)\,,\qquad 
\Szero_{t}:=\frac{p^{2}t}{2m}\,,\qquad 
\sigmaq_t := \sigma_0+\frac{i t}{2m\sigma_0}\,,
\end{equation*}
where $\xi_t $ is the solution of the free Hamilton equations, $\Szero_{t}$ is the (free) classical action, and $\sigmaq_t $ takes into account the spreading of the wave function.

If the dynamics is  free (i.e., $\alpha =0$), one has the identity
\begin{equation*}
\big(e^{-i\frac{t}{\hbar}H_{0}}\,\psi^{\hbar}_{\sigmazero,\xi}\big)(x) = e^{\frac{i}{\hbar}\Szero_{t}} \psi^{\hbar}_{\sigma_{t},\xi_t}(x)\,.
\end{equation*}
 The latter can be rewritten as   
\begin{equation}\label{exact}
\big(e^{-i\frac{t}{\hbar}H_{0}}\,\psi^{\hbar}_{\sigmazero,\xi}\big)(x) = e^{\frac{i}{\hbar}\Szero_{t}} \big(e^{itL_{0}} \phi^{\hbar}_{\sigma_{t},x}\big)(\xi)\,, 
\end{equation}
 where  $e^{itL_{0}}$ is the realization in $L^{\infty}(\RE^2)$ of the strongly continuous (in $L^{2}(\RE^{2})$) group of evolution generated by the self-adjoint operator $L_0 = \overline L$; explicitly, one has  $e^{itL_{0}}f(\xi) = f(\xi_t)$. 
 
Since $H_\alpha$ is a self-adjoint extension of $H^{\circ}_{0}:=H_{0}\!\!\upharpoonright\!\C^{\infty}_{c}(\RE\backslash\{0\})$, mimicking the identity \eqref{exact}, we compared $\big(e^{-i\frac{t}{\hbar}H_{\alpha}}\,\psi^{\hbar}_{\sigmazero,\xi}\big)(x)$ with $e^{\frac{i}{\hbar}\Szero_{t}} \big(e^{itL_{\beta}} \phi^{\hbar}_{\sigma_{t},x}\big)(\xi)$, with $L_{\beta}$ a self-adjoint extension   of $L^{\!\circ}_{0}:=L_{0}\!\!\upharpoonright\! \C^{\infty}_{c}({\mathscr M}_{0})$, ${\mathscr M}_{0}:=\RE^{2} \,\backslash\, \{(0,p) \,|\, p\!\in\!\RE\}$. Here, $\beta$ is a real constant which parameterizes  the self-adjoint extension, and it turns out that the optimal choice is  $\beta = 2\alpha/\hbar$ (see \cite{CFP19} for the details).

In the same spirit, in the present work, we study the small $\hbar$ asymptotic of $e^{-i\frac{t}{\hbar}H_{K}}\,\bm{\Xi}^{\hbar}_{\sigma_0,\xi} $ where $H_K$ is the quantum Hamiltonian obtained as a self-adjoint realization of $-(2m)^{-1}\hbar^2 \Delta$ on the star-graph with Kirchhoff conditions  in the vertex, and $\bm{\Xi}^{\hbar}_{\sigma,\xi} $ resembles a coherent state concentrated  on one  edge of the graph (see Section \ref{ss:truncated} below for the precise definition). \\

In the following sections of the introduction we give the main definitions and results. Section \ref{s:2} and \ref{s:semiclass} contain a detailed description of the quantum and semiclassical dynamics on the star-graph respectively. In Section \ref{s:4} we give the proofs of Theorems \ref{t:dynamics} and \ref{t:waveoperators}. Section \ref{s:5} contains  some additional remarks and comments. In the appendix we give a proof of a technical result, namely an explicit formula for the wave operators for the pair (Dirichlet Laplacian, Neumann Laplacian) on the half-line. 

\subsection{Quantum dynamics on the star-graph}
By star-graph we mean a non-compact graph, with $n$ edges (or leads) and one vertex. Each edge can be identified with a half-line, the origins of the half-lines coincide and identify the only vertex of the graph. 

We recall that the Hilbert space associated to the star-graph is  $L^2(\GG) \equiv \oplus_{\ell = 1}^{n} L^2(\RE_{+})$,  with the natural scalar product and norm; in particular, for the $L^2$-norm we use the notation
\[
\|\bm{\psi}\|_{L^2(\GG)} \equiv \left(\int_{\GG} dx\,|\bm{\psi}(x)|^2 \right)^{1/2} := \left(\sum_{\ell =1}^n\int_0^\infty dx\,|\psi_\ell(x)|^2 \right)^{1/2}.
\]
If  $\bm{\psi}\in L^2(\GG)$,  $\psi_\ell\in L^2(\RE_+)$ is its $\ell$-th component with respect to  the decomposition $\oplus_{\ell = 1}^{n} L^2(\RE_{+})$.  In a similar way one can define the associated  Sobolev spaces; in particular,  we set $H^2(\GG)\equiv  \oplus_{\ell = 1}^{n} H^2(\RE_{+})$, with the natural scalar product and norm. 

We are primarily interested in the semiclassical limit of the quantum dynamics generated by the Kirchhoff Laplacian on the star-graph, which is the operator
\begin{gather}
\dom( \HK ) := \Big\{\bm\psi \in H^2(\GG)\big|\; \psi_{1}(0) = \ldots = \psi_{n}(0)\,,\; \sum_{\ell = 1}^n \psi'_{\ell}(0) = 0 \Big\}\,, \label{HK1} \\
\HK\bm{\psi} := - \,\frac{\hbar^2}{2m}\, \bm{\psi}''  \,; \label{HK2}
\end{gather}
here $\bm{\psi}''$ denotes the element of $L^2(\GG)$ with components $\psi''_\ell$,  and $\bm{\psi}(0)$ (resp., $\bm{\psi}'(0)$) the vector in $\CO^n$ with components $\psi_\ell(0)$ (resp., $\psi'_\ell(0)$).
Functions in $\dom( \HK )$ are said to satisfy Kirchhoff (or Neumann, or standard, or natural) boundary conditions. 

In the analysis of the semiclassical limit of the wave and scattering operators, we will have to fix a reference  dynamics on the star-graph. To this aim we will consider the operator $\HDop$  (see also the equivalent definition in Eqs. \eqref{Hdirichlet1} - \eqref{Hdirichlet2} below)
\begin{gather*}
\dom( \HDop ) := \Big\{\bm\psi \in H^2(\GG)\big|\; \psi_{1}(0) = \ldots = \psi_{n}(0) = 0\Big\}\,,  \\
\HDop\,\bm{\psi} := - \,\frac{\hbar^2}{2m}\, \bm{\psi}'' \,;
\end{gather*}
we remark that $\HDop$ can be understood as the direct sum of $n$ copies of the Dirichlet Hamiltonian on the half-line (see Section \ref{ss:DKdynamics}  below for further details). 

We recall that the quantum wave operators and the corresponding scattering operator on $L^{2}(\GG)$, are defined by
\begin{gather}
\OmQ^{\pm} := \slim_{t\to\pm\infty}e^{i {t \over \hbar} \HK} e^{- i {t \over \hbar} \HDop}\,, \label{Wqdef}\\
\SQ := (\OmQ^{+})^{*}\, \OmQ^{-}\,. \label{Sqdef}
\end{gather}

These operators can be computed explicitly (see Proposition \ref{p:QWO} and Remark \ref{r:QSO} below), and component-wisely for $\ell = 1,\dots, n$ they read as follows:
\[
\big(\OmQ^{\pm} \bm{\psi}\big)_\ell = \sum_{\ell'=1}^n \Big( \delta_{\ell,\ell'} - \frac1n (1\mp \mathcal F_c^* \mathcal F_s)\Big)\psi_{\ell'}\,,
 \]
where $ \mathcal F_s$ and $\mathcal F_c$ are the Fourier-sine and  Fourier-cosine  transforms respectively (see Eqs. \eqref{F-sine} and \eqref{F-cosine});
\begin{equation}\label{whatis1}
(\SQ\,\bm{\psi})_\ell = \sum_{\ell'=1}^n \Big(  \delta_{\ell,\ell'} - \frac{2}{n}\Big) \psi_{\ell'}\,.
\end{equation}

\subsection{Semiclassical dynamics on a star-graph\label{ss:1.2}}
The generator of the semiclassical dynamics on the star-graph is obtained as a self-adjoint realization of the differential operator $-\,i\,{p \over m}\,\frac{\partial }{\partial q}$ in $\oplus_{\ell = 1}^{n} L^{2}(\RE_{+}\!\times\RE)$, $(q,p) \in \RE_{+} \!\times \RE$. To recover it  we will make use of the method  to  classify the singular perturbations of self-adjoint operators developed by one of us in \cite{P01} (see also \cite{P08}). 

To do so, the first step is to identify a simple dynamics on the star-graph, more precisely its generator. We shall consider classical particles moving on the edges of the graph with elastic collision at the vertex. 

We start by considering  the dynamics of a classical particle on the half-line with elastic collision at the origin. We obtain its generator as a limiting case from our previous work \cite{CFP19} and denote it by $\LD$. We postpone the precise definition of $\LD$ to Section \ref{ss:classhalfline}. Here we just note few facts. 

$\LD: \dom(\LD)\subset L^2(\RE_+\!\times\RE)\to L^2(\RE_+\!\times\RE)$ is  self-adjoint and acts on elements of its domain as 
\[
 (\LD f)(q,p)  = -\,i\,{p \over m}\,\frac{\partial f }{\partial q}(q,p) \qquad  \text{for\; $(q,p) \in \RE_{+} \!\times \RE$}\,.
\]

For all $t \in \RE$, the action of the unitary evolution group associated to it is explicitly given by 
\begin{equation}\label{group}
\big(e^{i t \LD} f\big)(q,p) =
\left\{ 
\begin{aligned} 
&f\big(q+\frac{pt}m,p\big) \qquad &\text{if }\,q+\frac{pt}m >0\,, \\ 
& - f\big(-q-\frac{pt}m,-p\big) \qquad &\text{if }\,q+\frac{pt}m <0 \,.
\end{aligned}
\right. 
\end{equation}

The (trivial) classical dynamics of a particle on the star-graph with elastic collision at the vertex  can be defined in the following way. Denote by $\fv$ a function of the form 
\[
\fv \in \oplus_{\ell = 1}^{n} L^{2}(\RE_{+}\!\times\RE)\,, \qquad \fv(q,p) \equiv \begin{pmatrix} f_1(q,p) \\ \vdots \\ f_n(q,p) \end{pmatrix}.
\]
If $\|\fv\|_{\oplus_{\ell = 1}^{n} L^{2}(\RE_{+}\!\times\RE)} =1$, $|f_{\ell}(q,p)|^2dqdp$ can be interpreted as the probability of finding a particle on the $\ell$-th edge of the graph, with position  in the interval $[q,q+dq]$ and  momentum in the interval $[p,p+dp]$.

Define the operator  $\LD^{\oplus}:= \oplus_{\ell = 1}^{n} \LD$;  the  associated dynamics is generated by the unitary group $e^{i \LD^{\oplus} t} = \oplus_{\ell = 1}^{n} e^{i\LD t}$, and it is trivial in the sense that it can be fully understood in terms of the dynamics on the half-line described above. 

We consider the map 
\[
(\gamma^{\oplus}_{+} \fv)_\ell(p) : = \lim_{q\to 0^+} f_\ell(q,p)  \qquad \mbox{for\, $\ell = 1,\dots, n$}\,, 
\]
defined on sufficiently smooth functions (we refer to Section \ref{s:semiclass} for the details). This map can be extended to a continuous one on $\dom(\LD^{\oplus})$. The operator $\LD^{\oplus}\!\upharpoonright \ker(\gamma^{\oplus}_+)$ is symmetric; in Theorem \ref{t:rescl}, by using the approach developed in \cite{P01, P08} we identify a family of self-adjoint extensions. Among those we select the one that turns out to be useful to study the semiclassical limit of  $\exp(-i H_K t/ \hbar)$ and denote it by $L_K$.

 We postpone the precise definition of $L_K$ to Section \ref{ss:LK}, see, in particular, Remark \ref{r:LK}. Here we just give component-wisely the formula for the associated unitary group, for $\ell = 1,\dots,n$ and for all $t\in\RE$: 
\begin{equation}\label{expLKgroup}
\big(e^{i t L_K} \fv\big)_\ell(q,p) =
\left\{ 
\begin{aligned} 
&f_\ell\big(q+\frac{pt}m,p\big) \qquad &\text{if }\,q+\frac{pt}m >0\,, \\ 
& \sum_{\ell'=1}^n\Big(\frac{2}{n} - \delta_{\ell,\ell'}\Big) f_{\ell'}\big(-q-\frac{pt}m,-p\big) \qquad &\text{if }\,q+\frac{pt}m <0\,. 
\end{aligned}
\right. 
\end{equation}

We define the classical wave operators and the corresponding scattering operator
on $\oplus_{\ell=1}^{n} L^{2}(\RE_{+}\!\times\RE)$ by
\begin{equation}
\OmC^{\pm} := \lim_{t\to\pm\infty}e^{i t \LD^{\oplus}}e^{- i t L_{K}} \label{Wcldef}
\end{equation}
and
\begin{equation}
\SC := (\OmC^{+})^{*}\, \OmC^{-}\,. \label{Scldef}
\end{equation}

These operators can be computed explicitly (see Proposition \ref{prop:WScl} below), and component-wisely they read as follows for $\ell = 1,\dots,n$ ($\theta$ is the Heaviside step function):
\begin{gather}
\big(\OmC^{\pm}\, \bm{f}\big)_\ell = \sum_{\ell'=1}^n \Big( \delta_{\ell,\ell'} - \frac2n \theta(\mp p)\Big)\,f_{\ell'}\,; \nonumber \\
(\SC\,\fv)_\ell = \sum_{\ell'=1}^n\Big(  \delta_{\ell,\ell'} - \frac{2}{n}\Big)\, f_{\ell'}\,. \label{whatis2}
\end{gather}

\subsection{Truncated coherent states on the star-graph\label{ss:truncated}}
In general, there is no natural definition of a coherent state on a star-graph, neither there is a unique way to extend coherent states through the vertex. Since we are interested in initial states concentrated on one edge  of the graph, we introduce the following  class of initial states. 
We denote by $\widetilde{\psi}^{\hbar}_{\sigma,\xi}$ the non-normalized restriction of $\psi^{\hbar}_{\sigma,\xi}$ (see Eq. \eqref{initial}) to $\RE_{+}$, namely,
\begin{equation}
\widetilde{\psi}^{\hbar}_{\sigma,\xi} \in L^2(\RE_{+}) \,, \qquad\quad
\widetilde{\psi}^{\hbar}_{\sigma,\xi}(x) = \psi^{\hbar}_{\sigma,\xi}(x) \quad \mbox{for\, $x > 0$}\,. \label{psitil}
\end{equation}

On the graph we consider the quantum states defined as 
\begin{definition}[Quantum states]\label{d:quantumstates} Let $\sigma \in \CO$, with $\Re\sigma = \sigma_0>0$,  and $\xi = (q,p) \in \RE_{+}\!\times \RE$; consider any normalized function $\Xi^{\hbar}_{\sigma,\xi} \in L^2(\RE_{+})$, such that
\begin{equation}
\big\|\Xi^{\hbar}_{\sigma,\xi} - \widetilde{\psi}^{\hbar}_{\sigma,\xi}\big\|_{L^2(\RE_{+})} \leq C_0\, e^{- \,\eps\,{q^2 \over \hbar |\sigma|^2}} \qquad \mbox{for some $C_0,\eps > 0$}\,. \label{COS}
\end{equation}
We are primarily interested in quantum states on the star-graph of the form 
\begin{equation*}
\bm{\Xi}^{\hbar}_{\sigma,\xi} \in L^{2}(\GG)\,, \qquad \bm{\Xi}^{\hbar}_{\sigma,\xi} \equiv \left(\!\!\begin{array}{c} \Xi^{\hbar}_{\sigma,\xi} \vspace{0.05cm}\\ 0 \vspace{-0.15cm}\\ \vdots \\ 0 \end{array}\!\!\right) . 
\end{equation*}
\end{definition}

Correspondingly, we will consider the family of classical states 
\begin{definition}[Classical states]\label{d:classicalstates} For any $\sigma$, $\xi$, and $\Xi^{\hbar}_{\sigma,\xi}$ as in Definition \ref{d:quantumstates}, consider the function $\Sigma^{\hbar}_{\sigma,x} : \RE_{+}\!\times\RE \to \CO$ defined by
\begin{equation}\label{CS1}
\Sigma^{\hbar}_{\sigma,x}(\xi) := \Xi^{\hbar}_{\sigma,\xi}(x)\,.
\end{equation}
We will make use of the family of  classical states on the star-graph  given by 
\begin{equation*}
\bm{\Sigma}^{\hbar}_{\sigma,x} \in \oplus_{\ell = 1}^{n} L^{2}(\RE_{+}\!\times\RE)\,, \qquad \bm{\Sigma}^{\hbar}_{\sigma,x} \equiv \left(\!\!\begin{array}{c} \Sigma^{\hbar}_{\sigma,x} \vspace{0.05cm}\\ 0 \vspace{-0.15cm}\\ \vdots \\ 0 \end{array}\!\!\right) .
\end{equation*}
\end{definition}

\subsection{Main results}
Our first result concerns the semiclassical limit of the dynamics. 
\begin{theorem}\label{t:dynamics} Let $\sigma_0 > 0$, $\xi = (q,p) \in \RE_{+}\!\times\RE$ and consider any initial state of the form $\bm{\Xi}^{\hbar}_{\sigma_0,\xi}$, together with its classical analogue $\bm{\Sigma}^{\hbar}_{\sigma_0,x}$. Then, for all $t \in \RE$ there holds
\begin{equation}\label{th1}
\left(\int_{\GG} dx\,\left|e^{-i {t \over \hbar} \HK} \bm{\Xi}^{\hbar}_{\sigma_0,\xi}(x) - e^{\frac{i}{\hbar}\Szero_{t}} \big( e^{it L_{K}} \bm{\Sigma}^{\hbar}_{\sigma_t,x}\big)(\xi)\right|^{2}\right)^{1/2} \!\leq C_0\, e^{- \,\eps\,{q^2 \over \hbar \sigma_0^2}} + 2 C_0\, e^{- \,\eps\,{(q + pt/m)^{2} \over \hbar |\sigma_t|^2}}\! + \sqrt{2} \, e^{-\frac{q^{2}}{4\hbar \sigma_{0}^{2}}} .
\end{equation}
\end{theorem}
\begin{remark}Let $t_{coll}(\xi):= - {mq/p}$ be the classical collision time. Whenever $|t - t_{coll}|\leq \eta m \sigma_0 \sqrt\hbar/|p| $  for some  positive constant $\eta$ the second term on the right-hand side of Eq. \eqref{th1} is larger than $2C_0\,e^{-\eps\eta^2}$. 
\end{remark}

In the second part of our analysis we study the semiclassical limit of the wave operators and of the scattering operator.

\begin{theorem}\label{t:waveoperators}
 Let $\sigma_0 > 0$, $\xi = (q,p) \in \RE_{+}\!\times\!(\RE\backslash \{0\})$ and consider any state of the form $\bm{\Xi}^{\hbar}_{\sigma_0,\xi}$, together with its classical analogue $\bm{\Sigma}^{\hbar}_{\sigma_0,x}$. Then, there hold 
\begin{equation}\label{th2wo}
\left(\int_{\GG} dx\,\left|\OmQ^{\pm}\, \bm{\Xi}^{\hbar}_{\sigma_0,\xi}(x) - (\OmC^{\pm}\bm{\Sigma}^{\hbar}_{\sigma_0,x})(\xi)\right|^{2}\right)^{\!1/2}
\leq \sqrt{\frac{2}{n}} \,\bigg(\sqrt{2}\,C_0\, e^{- \,\eps\,{q^2 \over \hbar \sigma_0^2}} + e^{-\,\frac{q^{2}}{4\hbar \sigma_0^2}} + e^{-\,\frac{\sigma_0^2 p^2}{\hbar}}\bigg) \,,
\end{equation}
and 
\begin{equation}\label{th2so}
\SQ\, \bm{\Xi}^{\hbar}_{\sigma_0,\xi}(x) = (\SC \,\bm{\Sigma}^{\hbar}_{\sigma_0,x})(\xi)\,.
\end{equation}
\end{theorem}

Identity \eqref{th2so} is an immediate consequence of Eqs. \eqref{whatis1} and \eqref{whatis2}, and of the definitions of $\bm{\Xi}^{\hbar}_{\sigma_0,\xi}$ and $\bm{\Sigma}^{\hbar}_{\sigma_0,x}$.

\begin{remark} Eq.\,\eqref{th2wo} makes evident that $\OmQ^{\pm}\, \bm{\Xi}^{\hbar}_{\sigma_0,\xi}$ and $(\OmC^{\pm}\,\bm{\Sigma}^{\hbar}_{\sigma_0,(\cdot)})(\xi)$ are exponentially close (with respect to the natural topology of $L^2(\GG)$) in the semiclassical limit $\hbar \sigma_0^2\,/q^2,\hbar/\sigma_0^2\, p^2 \!\to\! 0^+$ for any $\xi = (q,p)$ with $q > 0$ and $p \neq 0$. \\
As a matter of fact, it can be proved that the relation \eqref{th2wo} remains valid also for $p = 0$ if one puts $(\OmC^{\pm}\,\bm{\Sigma}^{\hbar}_{\sigma_0,x})(q,0) = \bm{\Sigma}^{\hbar}_{\sigma_0,x}(q,0)$; the latter position appears to be reasonable and is indeed compatible with the computations reported in the proof of Proposition \ref{prop:WScl}. Nonetheless, since $\exp(-\,\sigma_0^2\, p^2\!/\hbar) = 1$ in this case, the resulting upper bound is of limited interest for what concerns the semi-classical limit.
To say more, for $p = 0$ and $\hbar \sigma_0^2\,/q^2$ (or $C_0$) small enough, by a variation of the arguments described in the proof of Theorem \ref{t:waveoperators} one can derive the lower bound
\begin{equation*}
\left(\int_{\GG} dx\,\left|\OmQ^{\pm}\, \bm{\Xi}^{\hbar}_{\sigma_0,\xi}(x) - (\OmC^{\pm}\bm{\Sigma}^{\hbar}_{\sigma_0,x})(\xi)\right|^{2}\right)^{\!1/2}
\geq \sqrt{\frac{2}{n}} \,\bigg(1 - e^{-\,\frac{q^{2}}{4\hbar \sigma_0^2}} - \sqrt{2}\,C_0\, e^{- \,\eps\,{q^2 \over \hbar \sigma_0^2}}\bigg) \,.
\end{equation*}
This shows that, as might be expected, the classical scattering theory does not provide a good approximation for the quantum analogue when $p = 0$. On the contrary, notice that Eq.\,\eqref{th1} ensures a significant control of the error for the dynamics at any finite time $t \in \mathbb{R}$ even for $p = 0$.
\end{remark}

\section{The quantum theory\label{s:2}}
\subsection{Dirichlet dynamics on the half-line.}

Let us first consider the free quantum Hamiltonian for a quantum particle of mass $m$ on the whole real line, defined as usual by
$$ \Hfree : H^2(\RE) \subset L^2(\RE) \to L^2(\RE)\,, \qquad \Hfree\, \psi := - \,{\hbar^2 \over 2m}\,\psi''\,, $$
together with the associated  free unitary group $\Ufree_{t} := e^{-i {t \over \hbar} \Hfree}$ ($t \in \RE$). Correspondingly, let us recall that for any $\psi \in L^2(\RE)$ we have
\begin{gather}
(\Ufree_{t}\, \psi)(x) = \sqrt{{m \over 2\pi i \hbar t}}\, \int_{\RE} \!dy\;e^{i {m \over 2\hbar} {(x-y)^2 \over t}}\,\psi(y)\,.\label{Ufree}
\end{gather}

Let us further introduce the Dirichlet Hamiltonian on the half-line $\RE_{+}$, defined as usual by
\begin{gather*}
\dom(\HD) := H^{1}_{0}(\RE_{+}) \cap H^{2}(\RE_{+})\,, \qquad \HD \psi := - \,{\hbar^2 \over 2m}\,\psi''\,,
\end{gather*}
and refer to the associated unitary group $\UD_{t} := e^{- i {t \over \hbar} \HD}$ ($t \in \RE$). As well known, the latter operator can be  expressed as
\begin{gather}
\UD_{t} =\, \mathcal{U}^{-}_{t} -\, \mathcal{U}^{+}_{t} \,, \label{Up}
\end{gather}
where, in view of the identity \eqref{Ufree}, we introduced the bounded operators on $L^2(\RE_{+})$ defined as follows for $\psi \in L^2(\RE_{+})$ and $x \in \RE_{+}$:
\begin{gather}
(\mathcal{U}^{\pm}_{t}\, \psi)(x) := \sqrt{{m \over 2\pi i \hbar t}}\, \int_{0}^{\infty} \!\!dy\;e^{i {m \over 2\hbar} {(x \pm y)^2 \over t}}\,\psi(y)\,. \label{UUpm}
\end{gather}

\begin{remark}\label{rem:RUpm}
 Let  us consider the bounded operator
\[
\Theta : L^2(\RE_{+}) \to L^2(\RE)\,, \qquad
 (\Theta\, \psi)(x) = \left\{ \begin{aligned}
 			     &\psi(x) \qquad & \mbox{if\, $x>0$}\,, \\ 
			     &0              & \mbox{if\, $x<0$}\,,
			     \end{aligned}\right.
\]
together with its adjoint 
$$ 
\Theta^{*} : L^2(\RE) \to L^2(\RE_{+})\,, \quad\qquad (\Theta^{*} \psi)(x) = \psi(x) \quad (x \in \RE_{+})\,.
$$
Namely, $\Theta$ gives the extension by zero to the whole real line $\RE$ of any function on $\RE_{+}$, while $\Theta^{*}$ is the restriction to $\RE_{+}$ of any function on $\RE$. Note that $\Theta$ is an isometry. In fact, $\Theta^{*}\, \Theta$ is the identity on $L^2(\RE_{+})$ and $\Theta\, \Theta^{*}$ is an orthogonal projector (but not the identity) on $L^2(\RE)$; more precisely, we have ($\theta$ is the Heaviside step function)
\begin{gather*}
\Theta^{*}\, \Theta = \unit \quad \mbox{on\, $L^2(\RE_{+})$}\,, \\
\Theta\, \Theta^{*} : L^2(\RE) \to L^2(\RE)\,, \quad\qquad (\Theta\, \Theta^{*} \psi)(x) = \theta(x)\, \psi(x) \quad (x \in \RE)\,.
\end{gather*}

To proceed let us consider the parity operator
\begin{equation*}
 P : L^2(\RE) \to L^2(\RE)\,, \qquad\quad (P\, \psi)(x) = \psi(-x) \quad (x \in \RE)\,.
 \end{equation*}
Of course $P$ is a unitary, self-adjoint involution which commutes with the free Hamiltonian $\Hfree$, i.e., 
$$ \Hfree\, P = P\, \Hfree\,. $$
Furthermore it can be checked by direct inspection that
\begin{equation*}
\ran(P\,\Theta) = \ker(\Theta^{*})
\end{equation*}

Using the bounded linear maps introduced above, one can express the operators defined in Eq. \eqref{UUpm} as follows:
\begin{gather}
\mathcal{U}^{-}_{t} = \Theta^{*} \Ufree_{t}\,\Theta\,, \qquad\quad
\mathcal{U}^{+}_{t} = \Theta^{*} P\, \Ufree_{t}\,\Theta = \Theta^{*} \Ufree_{t}\,P\,\Theta\,. \label{UUpmop}
\end{gather}
Recalling that $(\Ufree_{t})^{*} = \Ufree_{-t}$, the above relation allow us to infer 
\begin{gather*}
(\mathcal{U}^{-}_{t})^{*} = \mathcal{U}^{-}_{-t}\,, \qquad\quad
(\mathcal{U}^{+}_{t})^{*} = \mathcal{U}^{+}_{-t}\,.
\end{gather*}
Let us finally point out that, on account of the obvious operator norms $\|\Theta\| = \|\Theta^{*}\| = 1$, $\|\Ufree_t\| = 1$ and $\|P\| = 1$, from Eq. \eqref{UUpmop} it readily follows
\begin{equation*}
\|\,\mathcal{U}^{\pm}_{t}\| \leq 1\,. 
\end{equation*}
\end{remark}

\subsection{Dirichlet and Kirchhoff dynamics on the star-graph.\label{ss:DKdynamics}}

Let us now introduce the quantum Hamiltonian on the graph $\GG$, corresponding to Dirichlet boundary conditions at the vertex. This coincides with the direct sum of $n$ copies of the Dirichlet Hamiltonian $\HD$ on the half-line $\RE_{+}$, namely:
\begin{gather}
\dom(\HDop) := \oplus_{\ell = 1}^{n}\, \dom(\HD) = \big(\!\oplus_{\ell = 1}^{n}\! H^{1}_{0}(\RE_{+})\big) \cap H^{2}(\GG)\,, \label{Hdirichlet1}\\
\HDop := \oplus_{\ell = 1}^{n} \HD : \dom(\HDop) \subset L^2(\GG) \to L^2(\GG)\,, \qquad \HDop\, \bm{\psi} := - \,{\hbar^2 \over 2m}\! \left(\!\begin{array}{c} \psi''_{1} \\ \vdots \\ \psi''_{n} \end{array}\! \right). \label{Hdirichlet2}
\end{gather}
In view of the identity \eqref{Up}, it can be readily inferred that the corresponding unitary group $e^{- i {t \over \hbar} \HDop}$ ($t \in \RE$) can be expressed as
\begin{gather}
e^{- i {t \over \hbar} \HDop} =\, \oplus_{\ell = 1}^{n}\, \mathcal{U}^{-}_{t} -\, \oplus_{\ell = 1}^{n} \,\mathcal{U}^{+}_{t} \,, \label{UpG}
\end{gather}
where $\mathcal{U}^{\pm}_{t}$ is defined as in Eq. \eqref{UUpm}.

To proceed let us consider the Kirchhoff Hamiltonian on the graph $\GG$. This is defined as in Eqs. \eqref{HK1} \eqref{HK2}. In what follows we denote by $\MC$ the $n\times n$ matrix with components
\begin{equation}\label{barefoot}
(\MC)_{\ell,\ell'} := \delta_{\ell,\ell'} - \frac2n\,, \qquad \mbox{for $\ell,\ell' = 1,\dots,n$}\,.
\end{equation}
By a slight abuse of notation we use the same symbol to denote the operator in $L^2(\GG)$ defined by
\[
(\MC\, \bm{\psi})_\ell := \sum_{\ell'=1}^n(\MC)_{\ell,\ell'}\psi_{\ell'} \qquad \bm{\psi} \in L^2(\GG)\,.
\]

By arguments similar to those given in the proof of \cite[Thm.\,2.1]{ACFN11} (cf. also \cite{GOO} and \cite[Eq.\,(7.1)]{KosSch}) we get
\begin{gather}
e^{-i {t \over \hbar} \HK} = \oplus_{\ell = 1}^{n}\,\mathcal{U}^{-}_{t} - {\MC} \oplus_{\ell = 1}^{n} \mathcal{U}^{+}_{t} . \label{dynHF}
\end{gather}

\subsection{The quantum wave operators and scattering operator.}
Let us consider the wave operators and the corresponding scattering operator on $L^{2}(\GG)$  respectively defined in Eqs.\,\eqref{Wqdef} and \eqref{Sqdef}. 

Since $\HK$ has purely absolutely continuous spectrum $\sigma(\HK) = [0,\infty)$, we have that $\OmQ^{\pm}$ are unitary on the whole Hilbert space $L^2(\GG)$, i.e.,
\begin{equation*}
(\OmQ^{\pm})^* \OmQ^{\pm} = \unit\,,
\end{equation*}
which in turn ensures\footnote{Of course, the same identity \eqref{Omnorm} can be derived straightforwardly from the fact that $\OmQ^{\pm}$ are defined as strong limits of unitary operators.}
\begin{equation}\label{Omnorm}
\|\OmQ^{\pm}\| = 1\,.
\end{equation}

Let us define the unitary operators $\mathcal F_s : L^2(\RE_+) \to  L^2(\RE_+) $  and $\mathcal F_c : L^2(\RE_+) \to  L^2(\RE_+) $:
\begin{gather}
(\mathcal F_s \psi)(k) := - \,\frac{2i}{\sqrt{2\pi}} \int_0^\infty\!\! dx\; \sin (kx)\, \psi (x)\qquad (k \in\RE_+)\,; \label{F-sine}\\
(\mathcal F_c \psi)(k) := \frac2{\sqrt{2\pi}} \int_0^\infty\!\! dx\; \cos (kx)\, \psi (x)  \qquad (k \in \RE_+)\,. \label{F-cosine}
\end{gather}

The wave operators can be computed explicitly. To this aim one could use the results  from Weder \cite{Weder15} (see also references therein), with some modifications, since in \cite{Weder15} the  reference dynamics is given by the Hamiltonian with Neumann boundary conditions.  For the sake of completeness, we prefer to give an explicit derivation of the result, obtained by taking the limit $t\to\pm\infty$ on the unitary groups. We remark that in \cite{Weder15}  the formulae are obtained by using the  Jost functions. 

We have the following explicit formulae for the wave operators:

\begin{proposition}\label{p:QWO}
The quantum wave operators can be expressed as
\begin{equation}\label{house}
\OmQ^{\pm} = \frac12 \oplus_{\ell = 1}^{n} (\unit \pm \mathcal F_c^* \mathcal F_s) + \frac12\, {\MC} \oplus_{\ell = 1}^{n} \!(\unit \mp \mathcal F_c^* \mathcal F_s) \,.  
\end{equation}
\end{proposition}

\begin{proof}
By Eqs. \eqref{UpG} and \eqref{dynHF} we easily obtain the identity 
\begin{equation}\label{Wqexp}
\OmQ^{\pm} = \slim_{t\to\pm\infty}\left(\oplus_{\ell = 1}^{n} (\mathcal{U}^{-}_{-t}\, \mathcal{U}^{-}_{t} - \mathcal{U}^{-}_{-t}\, \mathcal{U}^{+}_{t}) + {\MC} \oplus_{\ell = 1}^{n} \!(\mathcal{U}^{+}_{-t}\, \mathcal{U}^{+}_{t}  - \mathcal{U}^{+}_{-t}\, \mathcal{U}^{-}_{t})\right)\,.
\end{equation}

Let us find more convenient expressions for the operators on the  right-hand side. Let $\psi \in L^2(\RE_+)$ and define 
\[\psi_e(x) := \left\{\begin{aligned}&\psi(x) &\qquad \mbox{if\, $x>0$} \\
&\psi(-x) &\qquad \mbox{if\, $x<0$}
                                  \end{aligned}  
                          \right. \;;
\qquad\quad 
                         \psi_o(x) := \left\{\begin{aligned}&\psi(x) &\qquad \mbox{if\, $x>0$} \\
				-&\psi(-x) &\qquad \mbox{if\, $x<0$}
 				\end{aligned}
 			\right. \;.
 \] 
So, in view of Remark \eqref{rem:RUpm}, we have $\Theta \psi  = \frac{\psi_e +\psi_o }{2}$.  Hence, see Eq. \eqref{UUpmop}, 
\[\begin{aligned}
\mathcal{U}^{-}_{-t}\, \mathcal{U}^{-}_{t} \psi
& =
\Theta^{*} \Ufree_{-t}\,\Theta\, \Theta^{*} \Ufree_{t}\,\Theta \psi  \\ 
& = \frac14\; \Theta^{*} \left[  \Ufree_{-t} (\Theta^{*} \Ufree_{t}\psi_e)_e+\Ufree_{-t}(\Theta^{*} \Ufree_{t}\psi_o)_e+\Ufree_{-t}(\Theta^{*} \Ufree_{t}\psi_e)_o+\Ufree_{-t} (\Theta^{*} \Ufree_{t} \psi_o)_o \right] \\
&\hspace{-0.5cm} \textrm{(since $\Ufree_{t}\psi_e$ is even and $\Ufree_{t}\psi_o$ is odd, $(\Theta^{*} \Ufree_{t}\psi_e)_e =\Ufree_{t}\psi_e $ and $(\Theta^{*}\Ufree_{t}\psi_o)_o =\Ufree_{t}\psi_o $ we get)} \\
& =  \frac14\; \Theta^{*} \left[  \psi_e+\Ufree_{-t}(\Theta^{*} \Ufree_{t}\psi_o)_e+\Ufree_{-t}(\Theta^{*} \Ufree_{t}\psi_e)_o+ \psi_o \right] \\ 
& = \frac12\;   \psi
+\frac14\; \Theta^{*} \left[  \Ufree_{-t}(\Theta^{*} \Ufree_{t}\psi_o)_e+\Ufree_{-t}(\Theta^{*} \Ufree_{t}\psi_e)_o \right]. 
\end{aligned}\]

On the other hand, 
\[\begin{aligned}
\mathcal{U}^{-}_{-t}\, \mathcal{U}^{+}_{t} \psi
& =
\Theta^{*} \Ufree_{-t}\,\Theta\, \Theta^{*}\,\Ufree_{t}\,P\,\Theta\,\psi  \\ 
& \textrm{(Since $P\Theta\psi=P(\psi_e + \psi_o)/2 =(\psi_e - \psi_o)/2$ we get)} \\
& = \frac14\; \Theta^{*} \left[  \Ufree_{-t} (\Theta^{*} \Ufree_{t}\psi_e)_e-\Ufree_{-t}(\Theta^{*} \Ufree_{t}\psi_o)_e+\Ufree_{-t}(\Theta^{*} \Ufree_{t}\psi_e)_o-\Ufree_{-t} (\Theta^{*} \Ufree_{t} \psi_o)_o \right] \\
& = \frac14\; \Theta^{*} \left[  \psi_e-\Ufree_{-t}(\Theta^{*} \Ufree_{t}\psi_o)_e+\Ufree_{-t}(\Theta^{*} \Ufree_{t}\psi_e)_o- \psi_o \right] \\ 
& = \frac14\; \Theta^{*} \left[  - \Ufree_{-t}(\Theta^{*} \Ufree_{t}\psi_o)_e+\Ufree_{-t}(\Theta^{*} \Ufree_{t}\psi_e)_o \right]. 
\end{aligned}\]

Hence
\[
( \mathcal{U}^{-}_{-t}\, \mathcal{U}^{-}_{t} - \mathcal{U}^{-}_{-t}\, \mathcal{U}^{+}_{t})\, \psi = \frac12\;   \psi + \frac12\; \Theta^{*} \Ufree_{-t}(\Theta^{*} \Ufree_{t}\psi_o)_e \,.
 \]

Recall that  $\UD_t$ is the unitary group generated by the Dirichlet  Laplacian on the half-line; its integral kernel is given by 
 \begin{equation*}
 \UD_t(x,y) = \Ufree_{t}(x-y) -  \Ufree_{t}(x+y)\qquad \mbox{for\, $x,y \in\RE_+$}\,.
  \end{equation*}
Moreover, let $\UN_t$ be the  unitary group generated by the Neumann Laplacian on the half-line; its integral kernel is given by  
 \begin{equation*}
\UN_t(x,y) = \Ufree_{t}(x-y) +  \Ufree_{t}(x+y)\qquad \mbox{for\, $x,y \in\RE_+$}\,.
  \end{equation*}
  
 Note that, for $x \in \RE_{+}$,
\[
( \Theta^{*} \Ufree_{t}\,\psi_o) (x) 
= \int_\RE\! dy\;  \Ufree_{t}(x-y) \psi_o (y) 
= \int_0^\infty\!\!  dy\, \big(\Ufree_{t}(x-y) -  \Ufree_{t}(x+y)\big)\, \psi (y) = (\UD_t \psi)(x)\,,
\]
similarly
 \[
( \Theta^{*} \Ufree_{t}\,\psi_e) (x) = \int_0^\infty\!\!  dy\,  \big(\Ufree_{t}(x-y) +  \Ufree_{t}(x+y)\big)\, \psi (y)  = (\UN_t \psi)(x)\,.
  \]
 Hence, 
 \[
\Theta^{*}  \Ufree_{-t}(\Theta^{*} \Ufree_{t}\,\psi_o)_e  = \Theta^{*}  \Ufree_{-t}( \UD_t \psi)_e  = \UN_{-t} \,\UD_t\, \psi\,, 
 \]
and 
\begin{equation} \label{cow1}
( \mathcal{U}^{-}_{-t}\, \mathcal{U}^{-}_{t} - \mathcal{U}^{-}_{-t}\, \mathcal{U}^{+}_{t}) \psi = \frac12\;   \psi + \frac12\; \UN_{-t}\, \UD_t\, \psi\,.
\end{equation}

A similar computation gives 
\[\begin{aligned}
\mathcal{U}^{+}_{-t}\, \mathcal{U}^{+}_{t} \psi
& = \Theta^{*} \Ufree_{-t}\,P\,\Theta\, \Theta^{*} \Ufree_{t}\,P\Theta \psi  \\ 
& = \frac14\; \Theta^{*} \left[  \Ufree_{-t} (\Theta^{*} \Ufree_{t}\psi_e)_e-\Ufree_{-t}(\Theta^{*} \Ufree_{t}\psi_o)_e- \Ufree_{-t}(\Theta^{*} \Ufree_{t}\psi_e)_o+\Ufree_{-t} (\Theta^{*} \Ufree_{t} \psi_o)_o \right] \\
& = \frac12\; \psi
-\frac14\; \Theta^{*} \left[  \Ufree_{-t}(\Theta^{*} \Ufree_{t}\psi_o)_e+\Ufree_{-t}(\Theta^{*} \Ufree_{t}\psi_e)_o \right],
\end{aligned}\]
and 
\[\begin{aligned}
\mathcal{U}^{+}_{-t}\, \mathcal{U}^{-}_{t} \psi
& =
\Theta^{*} \Ufree_{-t}\,P\, \Theta\, \Theta^{*}\,\Ufree_{t}\,\Theta\,\psi  \\
& = \frac14\; \Theta^{*} \left[  \Ufree_{-t} (\Theta^{*} \Ufree_{t}\psi_e)_e+\Ufree_{-t}(\Theta^{*} \Ufree_{t}\psi_o)_e-\Ufree_{-t}(\Theta^{*} \Ufree_{t}\psi_e)_o-\Ufree_{-t} (\Theta^{*} \Ufree_{t} \psi_o)_o \right] \\
& = \frac14\; \Theta^{*} \left[   \Ufree_{-t}(\Theta^{*} \Ufree_{t}\psi_o)_e- \Ufree_{-t}(\Theta^{*} \Ufree_{t}\psi_e)_o \right]. 
\end{aligned}\]

Hence
\begin{equation}\label{cow2}
( \mathcal{U}^{+}_{-t}\, \mathcal{U}^{+}_{t} - \mathcal{U}^{+}_{-t}\, \mathcal{U}^{-}_{t})\, \psi = \frac12\; \psi - \frac12\; \Theta^{*} \Ufree_{-t}(\Theta^{*} \Ufree_{t}\psi_o)_e =  \frac12 \;  \psi - \frac12\; \UN_{-t} \,\UD_t \psi\,.
\end{equation}
 
To compute the wave operator we have to evaluate the limits $\slim_{t \to \pm \infty} \UN_{-t}\, \UD_t$; the latter give the wave operators $\Omega_{ND}^\pm$ for the pair (Dirichlet Laplacian, Neumann Laplacian) on the half-line, which are computed in Proposition \ref{p:WND} and equal $ \pm\, \mathcal F_c^* \mathcal F_s$. This, together with Eqs. \eqref{Wqexp}, \eqref{cow1}, and \eqref{cow2} concludes the proof of identity \eqref{house}.
\end{proof} 

\begin{remark}\label{r:QSO}
Note that the wave operators do not depend on $\hbar$. Moreover the scattering operator is given by\footnote{The last identity in Eq.\,\eqref{SQMC} can be derived by a simple computation starting from Eq.\,\eqref{house}, recalling that $ \mathcal F_s$ and $\mathcal F_c$ are unitary operators, and noting that $\MC^{*} \!= \MC$, $\MC^2 = \unit$.}
\begin{equation}\label{SQMC}
\SQ = (\OmQ^{+})^{*}\, \OmQ^{-} = \MC\,,
\end{equation}
the same formula is written component-wisely in Eq.\,\eqref{whatis1}. The matrix  $\MC$ given here equals $-\sigma^{(v)}$, where $\sigma^{(v)}$ is the scattering matrix at the vertex, as defined in \cite[Def.\,2.1.1.]{BK} (see also \cite[Ex.\,2.1.7, p.\,41]{BK} and \cite[Ex.\,2.4]{KS-jpa99}). The minus sign is due to the fact that as reference Hamiltonian  we chose Dirichlet boundary conditions, instead of Neumann boundary conditions, see also \cite{KS-jpa99}.
\end{remark}

\section{The semiclassical theory \label{s:semiclass}}

\subsection{Classical dynamics on the half-line with elastic collision at the origin. \label{ss:classhalfline}}
We start by recalling  some basic definitions and results from \cite{CFP19}. Let $X_{0}(q,p)=(p/m,0)$ be the vector field associated to the free (classical) Hamiltonian of a particle of mass $m$ and consider the differential operator 
$$
L:\S'(\RE^{2})\to \S'(\RE^{2})\,,\qquad Lf:= - \,i\,X_{0}\cdot\nabla f
$$
in the space of tempered distributions $\S'(\RE^{2})$.

We denote by 
\begin{gather*}
L_{0}:\dom(L_{0})\subseteq L^{2}(\RE^{2})\to L^{2}(\RE^{2})\,,\qquad (L_{0}f)(q,p)  = -\,i\,{p \over m}\,\frac{\partial f }{\partial q}(q,p)\,, \\
\dom( L_{0}) := \big\{ f \in L^2(\RE^2) \;\big|\; Lf\in L^2(\RE^2) \big\} \,, 
\end{gather*}
the maximal realization of $L$ in $L^{2}(\RE^{2})$. Posing $R^{0}_z := (L_{0} -z)^{-1}$ for $z\in\CO\backslash\RE$ one has 
\begin{equation}\label{res}
(R^{0}_z f)(q,p) =\sgn(\Im z)\;\frac{i\,m}{|p|} \int_{\RE}\!  dq'\,  \theta((q-q')\,p \Im z)\,{e^{i m z (q-q')/p}}\, f(q',p)\,. 
\end{equation}
Moreover  the action of the  (free) unitary group $e^{ i t L_{0}}$ ($t \in \RE$) is given by 
\begin{equation*}
\big(e^{ i t L_{0}} f\big)(q,p) = f\Big(q + {p t \over m} , p\Big)\,.
\end{equation*}

For any $f \in \S(\RE^{2})$ we define the map  $(\gamma f)(p) := f(0,p)$. For a comparison with the results in \cite[Sec. 2]{CFP19}, recall that the map $\gamma$ can be equivalently defined as  $(\gamma f)(p) = {1 \over \sqrt{2\pi}}\int_{\RE}\! dk\; \tilde f(k,p)$ where $\tilde f(k,p)$ is the Fourier transform of $f(q,p)$ in the variable $q$. By \cite[Lem. 2.1]{CFP19}, the map $\gamma$ extends to a bounded operator $\gamma :\dom(L_{0}) \to  L^2(\RE, |p|\,dp)$, where $\dom(L_{0}) \subset L^2(\RE)$ is endowed with the graph norm. For any  $z\in\CO\backslash\RE$, we define the bounded linear map
$$
G_{z}:L^{2}(\RE,|p|^{-1}dp)\to L^{2}(\RE^{2})\,,\qquad G_{z}:=(\gamma\, R^0_{\bar z})^{*}
$$
(here $L^{2}(\RE,|p|^{-1}dp)$ and $L^{2}(\RE,|p|dp)$ are considered as a dual couple with respect to the duality induced by the scalar product in $L^{2}(\RE)$). An explicit calculation gives 
$$
(G_{z}u)(q,p) =  \theta(q\,p \Im z)\;\sgn(\Im z)\;\frac{i\,m}{|p|}\,{e^{i m z q/p}}\,u(p)\,.
$$

Next we consider  the  classical motion of a point particle of mass $m$ on the whole real line, with elastic collision at the origin. The generator of the dynamics, denoted by $L_\infty$,  is obtained as a limiting case, for $\beta\to\infty$,  of the operator $L_\beta$ defined in \cite{CFP19}. To this aim we set 
\begin{equation*}
\Lambda^{\infty}_{z}:L^{2}(\RE,|p|dp)\to L^{2}(\RE,|p|^{-1}dp) \,,\quad
(\Lambda^{\infty}_{z}u)(p):=i \sgn(\Im z)\, \frac{2|p|}{m}\, u(p)\,.
\end{equation*}
In addition, let us consider the projector on even functions (here either $\rho(p)=|p|$ or $\rho(p)=|p|^{-1}$)
\begin{equation*}
\Pi_{ev} : L^2(\RE,\rho dp) \to L^2(\RE,\rho dp)\,, \qquad (\Pi_{ev} f)(p) := {1 \over 2}\,\big( f(p) + f(-p) \big)\,.
\end{equation*}
By \cite[Thm. 2.2]{CFP19}, here employed with $\beta \to \infty$, the operator $L_\infty$ is defined by 
\begin{gather*}
\dom(L_{\infty}):= \big\{f\in L^{2}(\RE^{2}) \;\big|\; f=f_{z}+G_{z}\,\Lambda^{\infty}_{z}\,\Pi_{ev}\gamma f_{z}\,,\; f_{z}\in\dom(L_{0})\big\}\,, \nonumber \\
(L_{\infty}-z)f=(L_{0}-z)f_{z}\,, \label{Linf}
\end{gather*}
for all $z\in\CO\backslash\RE$.
The associated resolvent operator $R^{\infty}_z := (L_{\infty} - z)^{-1}$ ($z \in \CO \backslash \RE$) can be expressed as follows, in terms of the free resolvent $R^{0}_z $ and of the trace operator $\gamma$:
\begin{equation}
(R^{\infty}_z f)(q,p) = (R^{0}_z f)(q,p) - \,\theta(q p \Im z)\,e^{i m z q/ p}\,\big[(\gamma R^{0}_z f)(p) + (\gamma R^{0}_z f)(-p)\big]\,. \label{ResDir}
\end{equation}
More explicitly, we have
\begin{align}
& (R^{\infty}_{z} f)(q,p) 
= \sgn(\Im z)\,\frac{i\,m}{|p|} \int_{\RE}\!\!dq' \Big[\theta\big((q-q')\,p \Im z\big)\;{e^{i m z (q-q')/p}}\, f(q',p) \nonumber \\
& \qquad - {\theta(q p \Im z) } \Big(
\theta(-q' p \Im z)\, e^{i m z (q-q')/p}\, f(q',p) 
+ \theta(q' p \Im z)\,e^{i m z (q+q')/p}\, f(q',-p) \Big)\Big]\,. \label{ResDir2}
\end{align}

Correspondingly, let us recall that \cite[Prop. 2.4]{CFP19} gives, for all $t \in \RE$,
\begin{equation}
\big(e^{i t L_{\infty}} f\big)(q,p) = \big(e^{i t L_{0}} f\big)(q,p) - \theta(-t q p)\,\theta\!\left({|p t| \over m} - |q|\right)\!\Big[\big(e^{i t L_{0}} f\big)(q,p) + \big(e^{i t L_{0}} f\big)(-q,-p) \Big] \label{dinLinf}
\end{equation}
(here, for a comparison with \cite{CFP19},  we used $\big(e^{i t L_{0}} f(-\,\cdot\,,-\,\cdot\,)\big)(q,p) = \big(e^{ i t L_{0}} f\big)(-q,-p)$\,). 

To proceed, let us introduce the lateral traces defined by
\begin{align}
\big(\gamma_{+}(R^{\infty}_{z} f)\big)(p)
& :=  (\gamma R^{0}_z f)(p) - \,\theta(p \Im z)\,\big[(\gamma R^{0}_z f)(p) + (\gamma R^{0}_z f)(-p)\big] \nonumber \\
& \,= \theta(-p \Im z)\, (\gamma R^{0}_z f)(p) - \theta(p \Im z)\, (\gamma R^{0}_z f)(-p)\,, \label{trp}\\
\big(\gamma_{-}(R^{\infty}_{z} f)\big)(p)
& :=  (\gamma R^{0}_z f)(p) - \,\theta(-p \Im z)\,\big[(\gamma R^{0}_z f)(p) + (\gamma R^{0}_z f)(-p)\big] \nonumber \\
& \,= \theta(p \Im z)\, (\gamma R^{0}_z f)(p) - \theta(-p \Im z)\, (\gamma R^{0}_z f)(-p) \label{trm}
\end{align}
(here we used the trivial identity $\theta(-s) = 1- \theta(s)$\,). Clearly,  $\gamma_{\pm}(R^{\infty}_{z} f)$ are odd functions and, using again \cite[Lem. 2.1]{CFP19},
\begin{equation}
\gamma_{\pm} : \dom(L_{\infty}) \to L^2_{odd}(\RE,|p|dp) \label{trlatdef}
\end{equation}
is a bounded operator. We remark that the action of  $\gamma_\pm$ can be understood as $\big(\gamma_{\pm}(R^{\infty}_{z} f)\big)(p)=\lim_{q\to0^\pm} (R^{\infty}_{z} f)(q,p)$.

For the subsequent developments it is convenient to express the free resolvent $R^{0}_z$ in terms of $R^{\infty}_z$. More precisely, we have the following explicit characterization.

\begin{lemma}\label{lem:ResDir} For any $z \in \CO \backslash \RE$ and for any $f \in L^2(\RE^2)$, there holds
\[
(R^{0}_z f)(q,p) = 
(R^{\infty}_z f)(q,p) - \theta(q p \Im z)\,\sgn(q)\,e^{i m z q/p} \, \Big[ (\gamma_{+}R^{\infty}_z f)(p) - (\gamma_{-}R^{\infty}_z f)(p)\Big]\,. 
\]
\end{lemma}

\begin{proof}
From Eqs. \eqref{trp} and \eqref{trm} we readily infer that
\begin{align*}
\theta(- p \Im z)\, (\gamma_{+} R^{\infty}_z f)(p) 
& = \theta(- p \Im z)\, (\gamma R^{0}_z f)(p) \,, \\
\theta(p \Im z)\, (\gamma_{-} R^{\infty}_z f)(p) 
& = \theta(p \Im z)\, (\gamma R^{0}_z f)(p)
\,.
\end{align*}
The above relations imply, in turn,
\begin{equation*}
(\gamma R^{0}_z f)(p)
 = \theta(- p \Im z)\, (\gamma_{+} R^{\infty}_z f)(p)
+ \theta(p \Im z)\, (\gamma_{-} R^{\infty}_z f)(p) \,,\label{gamma1}
\end{equation*}
and, since $\gamma_{\pm} R^{\infty}_z f$ are  odd functions, 
\begin{equation*}
(\gamma R^{0}_z f)(-p) 
 = -\,\theta(p \Im z)\, \big(\gamma_{+} R^{\infty}_z f\big)(p)
-\, \theta(-p \Im z)\, \big(\gamma_{-} R^{\infty}_z f\big)(p)\,. \label{gamma2}
\end{equation*}
Substituting the latter identities into Eq. \eqref{ResDir} and noting that
$$ \theta(q p \Im z) = \big(\theta(q) + \theta(-q)\big)\,\theta(q p \Im z) = \theta(q)\theta(p \Im z) + \theta(-q)\theta(- p \Im z)\,, $$
and
$$ \sgn(q)\theta(q p \Im z)  = \theta(q)\theta(p \Im z) - \theta(-q)\theta(- p \Im z)\,, $$
we obtain
\begin{align*}
& (R^{\infty}_z f)(q,p) \\
& = (R^{0}_z f)(q,p) - e^{i m z q/p} \Big[
\theta(q)\,\theta(p \Im z)\,(\gamma_{-} R^{\infty}_z f)(p)
- \theta(q)\,\theta(p \Im z)\,(\gamma_{+} R^{\infty}_z f)(p)
\\
& \hspace{4cm}+ 
\theta(-q)\,\theta(-p \Im z)\,(\gamma_{+} R^{\infty}_z f)(p)
- \theta(-q)\,\theta(- p \Im z)\,(\gamma_{-} R^{\infty}_z f)(p)
\Big] \\
& = (R^{0}_z f)(q,p) +\,\theta(q p \Im z)\,\sgn(q)\,e^{i m z q/p} \, \Big[
(\gamma_{+} R^{\infty}_z f)(p) - (\gamma_{-} R^{\infty}_z f)(p)\Big]\,,
\end{align*}
which suffices to infer the thesis.
\end{proof}

Similarly, for the unitary operator describing the dynamics we have

\begin{lemma}\label{lem:UDir} For any $t \in \RE$ and for any $f \in L^2(\RE^2)$, there holds
\[
\big(e^{i t L_{0}} f\big)(q,p) = \big(e^{i t L_{\infty}} f\big)(q,p) - \theta(-t q p)\,\theta\!\left({|p t| \over m} - |q|\right)\!\Big[\big(e^{i t L_{\infty}} f\big)(q,p) + \big(e^{i t L_{\infty}} f\big)(-q,-p) \Big]\,. 
\]
\end{lemma}

\begin{proof} First note that the identity in Eq. \eqref{dinLinf} entails
$$
\big(e^{i t L_{\infty}} f\big)(-q,-p) = \big(e^{i t L_{0}} f\big)(-q,-p) - \theta(-t q p)\,\theta\!\left({|p t| \over m} - |q|\right)\!\Big[\big(e^{i t L_{0}} f\big)(q,p) + \big(e^{i t L_{0}} f\big)(-q,-p) \Big]\,. 
$$
From the above relation and the previously cited equation, we derive
\begin{align*}
& \big(e^{i t L_{\infty}} f\big)(q,p) + \big(e^{i t L_{\infty}} f\big)(-q,-p) \\
& = \left(1 - 2\,\theta(-t q p)\,\theta\!\left({|p t| \over m} - |q|\right)\!\right)\!\Big[\big(e^{i t L_{0}} f\big)(q,p) + \big(e^{i t L_{0}} f\big)(-q,-p) \Big]\,,
\end{align*}
which in turn implies 
\begin{align*}
& \theta(-t q p)\,\theta\!\left({|p t| \over m} - |q|\right)\! \Big[\big(e^{i t L_{\infty}} f\big)(q,p) + \big(e^{i t L_{\infty}} f\big)(-q,-p) \Big] \\
& = - \,\theta(-t q p)\,\theta\!\left({|p t| \over m} - |q|\right)\!\Big[\big(e^{i t L_{0}} f\big)(q,p) + \big(e^{i t L_{0}} f\big)(-q,-p) \Big]\,.
\end{align*}
The thesis follows upon substitution of the above identity into Eq.\,\eqref{dinLinf}.
\end{proof}

Let us now consider the natural decomposition\footnote{All the equalities in Eq.\,\eqref{L2R2} must be understood as isomorphisms of Hilbert spaces (see, e.g., \cite[p. 85]{HS2012}).}
\begin{align}
& L^2(\RE^2) \equiv L^2\big((\RE_{-} \cup \RE_{+})\!\times\! \RE,dqdp\big) = \big(L^2(\RE_{-},dq) \oplus L^2(\RE_{+},dq) \big) \otimes L^2(\RE,dp) \label{L2R2}\\
& = \big(L^2(\RE_{-},dq) \otimes L^2(\RE,dp)\big) \oplus \big(L^2(\RE_{+},dq) \otimes L^2(\RE,dp)\big)
= L^2(\RE_{-}\! \times \RE) \oplus L^2(\RE_{+}\! \times \RE)\,, \nonumber
\end{align}
and notice that both the subspaces $L^2(\RE_{-}\! \times \RE) \equiv L^2(\RE_{-}\! \times \RE) \oplus \{0\}$ and $L^2(\RE_{+}\! \times \RE) \equiv \{0\} \oplus L^2(\RE_{+}\! \times \RE)$ are left invariant by the resolvent $R^{\infty}_z$, this is evident from Eq. \eqref{ResDir2}. Taking this into account, we introduce the bounded operator 
\begin{equation}\label{Rplusz}
\RD_{z} : L^{2}(\RE_{+}\!\times\RE)\to L^{2}(\RE_{+}\!\times\RE)\,, \qquad
(\RD_{z} f)(q,p) := \big(R^{\infty}_{z}(0\oplus f)\big)(q,p)\,.
\end{equation}
By direct computations, from Eq.\,\eqref{ResDir2} (here employed with $q > 0$) we get
\begin{align*}
& \big(\RD_{z} f\big)(q,p) \\
& =
\sgn(\Im z)\,\frac{i\,m}{|p|} \int_{0}^{\infty}\!\!\!dq' \Big[\theta\big((q-q')p \Im z\big)\,{e^{i m z (q-q')/p}}\, f(q',p) - \theta(p \Im z)\,e^{i m z (q+q')/p}\, f(q',-p)\Big] \\
& = \sgn(\Im z)\,\frac{i\,m}{|p|}\,e^{imzq/p}\Bigg[ \theta(p\Im z)\bigg(\int_{0}^{q}\!dq' \,{e^{-i m zq'/p}}\, f(q',p)  -\int_{0}^{\infty}\!\!\! dq'\, e^{i m z q'/p}\, f(q',-p)\bigg) \\
&\hspace{9cm} + \theta(-p\Im z)\int_{q}^{\infty}\!\!\!dq'\, {e^{-i m zq'/p}}\, f(q',p)\Bigg]\,.
\end{align*}

We denote by $\LD$ the self-adjoint operator in $L^{2}(\RE_{+}\!\times\RE)$ having $\RD_{z}$ as resolvent, so that 
\begin{equation}
\dom(\LD)=\big\{f \in L^{2}(\RE_{+}\!\times\RE)\,\big|\,(0\oplus f)\in \dom(L_{\infty})\big\}\,,\qquad 
\LD f :=L_{\infty}\,(0\oplus f)\,. \label{Lpdef}
\end{equation}
For all $(q,p) \!\in\! \RE_{+}\!\times\! \RE$, $t \!\in\! \RE$ and $f \!\in\! L^2(\RE_{+}\!\times\! \RE)$, from the above definition and from Eq.\,\eqref{dinLinf} we get\footnote{Note that for $q > 0$ we have $ \theta(-t q p)\,\theta\left({|pt| \over m} - |q|\right) = \theta(-t p)\,\theta\left(-{pt \over m} - q\right) = \theta\left(-q -{pt \over m}\right) = 1 - \theta\left(q +{pt \over m}\right)$.}
\begin{align}
\big(e^{i t \LD} f\big)(q,p) 
& = \big(e^{i t L_{\infty}} (0 \oplus f)\big)(q,p) \nonumber \\
& = \theta\Big(q +{pt \over m}\Big) \big(e^{i t L_{0}} (0\oplus f)\big)(q,p) - \theta\Big(\!-q -{pt \over m}\Big) \big(e^{i t L_{0}} (0\oplus f)\big)(-q,-p)\,, \label{dynClaDir}
\end{align}
which describes the motion of a classical particle on the half-line $\RE_{+}$ with elastic collision at $q = 0$.
Let us also mention that, in view of the basic identity
\begin{gather}
\big(e^{it L_{0}} (0\oplus f)\big)(q,p) = \theta\!\left(\!q + {p t \over m}\!\right)\! \big(e^{it L_{0}} (0 \oplus f)\big)(q,p)\,, \label{idfree1}
\end{gather}
the above relation \eqref{dynClaDir} is equivalent to
\begin{align}
\big(e^{i t \LD} f\big)(q,p) 
& = \big(e^{i t L_{0}} (0\oplus f)\big)(q,p) - \big(e^{i t L_{0}} (0\oplus f)\big)(-q,-p)\,. \label{dynClaDir2}
\end{align}
Another equivalent (and more explicit) formula for the action of the unitary group $e^{i t \LD}$ is the one given in Eq. \eqref{group}. 

Finally, from Lemmata \ref{lem:ResDir} and \ref{lem:UDir} (here employed with $q > 0$) we derive, respectively,
\begin{gather}
\big(R^{0}_z (0\oplus f)\big)(q,p) 
= (\RD_{z} f)(q,p) - \theta(p \Im z)\,e^{i m z q/p} \,(\gamma_{+}\RD_{z} f)(p) \,, \nonumber 
\\
\big(e^{i t L_{0}} (0\oplus f)\big)(q,p) 
= \theta\!\left(q + {p t \over m}\right)\! \big(e^{i t \LD} f\big)(q,p) - \theta\!\left(- q -\,{p t \over m}\right)\! \big(e^{i t \LD} f\big)(-q,-p)\,. \label{eitL0}
\end{gather}

\subsection{Classical dynamics on the graph with total reflection at the vertex.}
Let us now consider the ``classical'' Hilbert space $\oplus_{\ell = 1}^n L^2(\RE_{+}\! \times \RE)$ and indicate any of its elements with the vector notation
$$
\fv \in \oplus_{\ell = 1}^{n} L^{2}(\RE_{+}\!\times\RE)\,, \qquad \fv(q,p) \equiv \left(\!\!\begin{array}{c} f_1(q,p) \\ \vdots \\ f_n(q,p) \end{array}\!\!\right).
$$

Let $\LD$ be defined according to Eq.\,\eqref{Lpdef}, and consider the classical dynamics on the star-graph $\GG$ with total elastic collision in the vertex; this is described by the self-adjoint operator 
$$ \LD^{\oplus}:= \oplus_{\ell = 1}^{n} \LD \,. $$
The associated resolvent and time evolution operators are respectively given by
\begin{gather*}
R^{D\oplus}_{z} := (\LD^{\oplus} - z)^{-1} = \oplus_{\ell = 1}^{n} \RD_{z} \,, \label{Rpdef} \\
e^{i t \LD^{\oplus}} = \oplus_{\ell = 1}^{n} e^{i t \LD} \,. \nonumber
\end{gather*}

Explicitly, for $\ell=1,\dots,n$ and $t\in\RE$, from Eq. \eqref{group} we derive
\[
\big(e^{i t \LD^{\oplus}} \fv\big)_\ell(q,p) =
\left\{ 
\begin{aligned} 
&f_\ell\big(q+\frac{pt}m,p\big) \qquad &\text{if }\,q+\frac{pt}m >0\,, \\ 
& - f_\ell\big(-q-\frac{pt}m,-p\big) \qquad &\text{if }\,q+\frac{pt}m <0 \,.
\end{aligned}
\right. 
\]

\subsection{Singular perturbations of the classical dynamics on the graph. \label{ss:LK}}
Let us consider the restriction to $\dom(\LD)$ of the trace map $\gamma_{+}$ introduced in Eq.\,\eqref{trlatdef}; this defines a bounded operator
$$ \gamma_{+} : \dom(\LD) \to L^2_{odd}(\RE,|p|dp)\,. $$
We use the above map to define a trace operator on the graph:
$$ \gamma^{\oplus}_{+} := \oplus_{\ell = 1}^{n} \gamma_{+} : \oplus_{\ell = 1}^n \dom(\LD) \to \oplus_{\ell=1}^n L^2_{odd}(\RE,|p|dp)\,. $$

In what follows we use the technique developed by one of us in  \cite{P01, P08} to characterize all the self-adjoint extensions of the symmetric operator $\LD^{\oplus}\!\!\upharpoonright\!\ker(\gamma^{\oplus}_+)$ (see Theorem \ref{t:rescl} below). Among those we select the one that turns out to be useful to study the semiclassical limit of  $\exp(-i H_K t/ \hbar)$, see Remark \ref{r:LK}. 

To proceed, we introduce the operator 
$$ 
\Gc^{\oplus}_z := \gamma^{\oplus}_{+} R^{D\oplus}_z \equiv \oplus_{\ell = 1}^n \Gc^{+}_z : \oplus_{\ell}^n L^2(\RE_{+}\!\times\RE) \to \oplus_{\ell=1}^n L^2_{odd}(\RE,|p|dp)\,, \qquad \Gc^{+}_z := \gamma_{+} \RD_z\,,
$$
and its adjoint with complex conjugate parameter, 
$$ 
G^{\oplus}_z := \big(\Gc^{\oplus}_{\bar{z}}\big)^* \equiv \oplus_{\ell = 1}^n G^{+}_z :  \oplus_{\ell}^n L^2_{odd}(\RE,|p|dp) \to \oplus_{\ell=1}^n L^2(\RE_{+}\!\times\RE)\,, \qquad G^{+}_z := \big(\Gc^{+}_{\bar{z}}\big)^*\,.
$$
Note that by Eq. \eqref{res} one has 
\[
\big(\gamma R^{0}_z (0 \oplus f)\big)(p)  = \sgn(\Im z) \frac{i\,m}{|p|}\int_{0}^\infty \!  dq'\,  \theta(-p \Im z)\,{e^{ - i m z q'/p}}\, f(q',p)\,. 
\]
From the latter identity, together with Eqs. \eqref{Rplusz} and \eqref{trp}, we derive
\begin{align}
& (\Gc^{+}_z f)(p) = \big(\gamma_{+} R^{\infty}_z (0 \oplus f)\big)(p) \label{Gcexp}\\
& = \theta(-p \Im z)\, \big(\gamma R^{0}_z (0 \oplus f)\big)(p) - \theta(p \Im z)\,\big(\gamma R^{0}_z (0 \oplus f)\big)(-p)\,, \nonumber \\
& = \sgn(Im z)\,{i m \over |p|} \int_{0}^{\infty}\! dq' \left(\theta(-p \Im z)\, e^{-i m z q'/p} \, f(q',p) - \theta(p \Im z)\, e^{i m z q'/p} \, f(q',-p)\right). \nonumber
\end{align}
In view of the latter expression, for all $f \in L^2(\RE_{+}\!\times\RE)$ and any $\phi \in L^2_{odd}(\RE,|p|dp)$ we have
\begin{align*}
& \int_{\RE_{+}\!\times\RE}\!\!\!\!dqdp\, \overline{(G^{+}_z \phi)(q,p)}\,f(q,p)
= \int_{\RE}\!\!dp\, \overline{\phi(p)}\, (\Gc^{+}_{\bar{z}}f)(p) \\
& = \int_{\RE}\!\!dp\, \overline{\phi(p)} \left(\sgn(Im \bar{z})\,{i m \over |p|} \int_{0}^{\infty}\! dq \left(\theta(-p \Im \bar{z})\, e^{-i m \bar{z} q/p} \, f(q,p) - \theta(p \Im \bar{z})\, e^{i m \bar{z} q/p} \, f(q,-p)\right)\right) \\
& = \int_{\RE_{+}\!\times \RE}\!\!\!\!dqdp\, \overline{\big(\phi(p) - \phi(-p)\big)} \;\theta(p Im z)\,\sgn(Im z)\left(-\,{i m \over |p|}\right) e^{-i m \bar{z} q/p}\, f(q,p) \\
& = \int_{\RE_{+}\!\times \RE}\!\!\!\!dqdp\, \overline{\theta(p Im z)\,\sgn(Im z)\,{2i m \over |p|}\,e^{i m z q/p}\, \phi(p)} \; f(q,p)\,,
\end{align*}
which proves that
\begin{equation}
(G^{+}_z \phi)(q,p) = 2\,g_z(q,p)\, \phi(p)\,, \qquad
g_z(q,p) := \theta(p Im z)\,\sgn(Im z)\,{i m \over |p|}\,e^{i m z q/p}\,. \label{gzdef}
\end{equation}
On account of the identities 
$$
\big(\gamma_{+}(g_{z}-g_{w})\big)(p)=\frac{i\,m}{2\,|p|}\,\big(\sgn(\Im z)-\sgn(\Im w)\big)\,,
$$
which can be easily checked by a direct calculation (see also \cite[p. 7]{CFP19}), 
and
$$
\gamma_{+}(G^{+}_{z}-G^{+}_{w})=(z-w)\,(G^{+}_{\bar w})^{*}G^{+}_{z}\,,
$$
which is consequence of the first resolvent identity (see \cite[Lem. 2.1]{P01}, paying attention to the different sign convention in the definition of the resolvent),
we have that the linear map
\begin{gather}
\dom(M^{+}_{z}):=L_{odd}^{2}(\RE,|p|^{-1}dp) \cap L_{odd}^{2}(\RE,|p|dp)\,, \nonumber \\
M^{+}_{z} : \dom(M^{+}_{z})\subset L_{odd}^{2}(\RE,|p|^{-1}dp) \to L_{odd}^{2}(\RE,|p|dp)\,, \nonumber \\
(M^{+}_{z}\phi)(p) := 2\, m^{\infty}_{z}(p)\,\phi(p)\,, \qquad 
m^{\infty}_{z}(p) := -\,\sgn(\Im z)\,\frac{i\,m}{2|p|}\,, \label{mzdef}
\end{gather}
satisfies the identities 
$$
(M_{z}^{+})^{*}=M_{\bar z}^{+}\,,\qquad  M_{z}^{+}-M_{w}^{+}=(w-z)(G^{+}_{\bar w})^{*}G^{+}_{z}\,.
$$
Hence, setting 
$$
M^{\oplus}_{z}:\oplus_{\ell=1}^{n}\dom(M^{+}_{z})\subset \oplus_{\ell=1}^{n}L_{odd}^{2}(\RE,|p|^{-1}dp) \to\oplus_{\ell=1}^{n} L_{odd}^{2}(\RE,|p|dp)\,,\quad M^{\oplus}_{z}:=\oplus_{\ell=1}^{n}M_{z}^{+}\,,
$$
one gets the identities
$$
(M_{z}^{\oplus})^{*}=M_{\bar z}^{\oplus}\,,\qquad  M_{z}^{\oplus}-M_{w}^{\oplus}=(w-z)(G^{\oplus}_{\bar w})^{*}G^{\oplus}_{z}\,.
$$
To proceed, let us consider any orthogonal projector $\PP:\CO^{n}\to\CO^{n}$ and any self-adjoint operator  $\BB:\CO^{n}\to\CO^{n}$,  represented by the matrices with components $(\PP)_{\ell,\ell'}$ and $(\BB)_{\ell,\ell'}$ respectively. By a slight abuse of notation we use the same symbols to denote the corresponding operators on vector valued functions; e.g., for $\fv = \oplus_{\ell=1}^n f_\ell\in \oplus_{\ell=1}^{n} L^{2}(\RE_{+}\!\times\RE)$, one has $\PP \,\fv  \in \oplus_{\ell=1}^{n} L^{2}(\RE_{+}\!\times\RE)$ with components ($\ell = 1,\dots,n$)
\[
(\PP\, \fv)_\ell  = \sum_{\ell'=1}^{n} (\PP)_{\ell \ell'}f_{\ell'}\,, \]
or, for $\oplus_{\ell=1}^{n}\phi_{\ell} \in \oplus_{\ell=1}^{n}L_{odd}^{2}(\RE,\rho dp)$, one has $\PP(\oplus_{\ell=1}^{n}\phi_{\ell}):=\oplus_{\ell=1}^{n}\Big(\sum_{j=1}^{n} (\PP)_{\ell \ell'}\phi_{\ell'}\Big)$, and similarly for $\BB$. 
Then, by \cite[Thm. 2.1]{P01} here employed with $\tau:=\PP \gamma^{\oplus}_{+}$, we obtain the following

\begin{theorem}\label{t:rescl}
Let $z\in\CO\backslash \RE$. Assume that $\PP \, M_z^\oplus = M_z^\oplus \, \PP$ and  $\PP \, \BB = \BB \, \PP$.  Then, the linear bounded operator
\begin{equation}
R^{\PP,\BB}_{z}:= R^{D\oplus}_{z} \!+ G^{\oplus}_{z}\,\PP \,\big(\BB + M^{\oplus}_{z}\big)^{-1} \PP\, \Gc^{\oplus}_{z} \label{ResClass}
\end{equation} 
is the resolvent of a self-adjoint extension $L_{\PP,\BB}$ of the densely defined, closed symmetric operator $\LD^{\oplus} \!\upharpoonright \ker(\gamma^{\oplus}_{+})$. Such an extension acts on its domain
$$
\dom(L_{\PP, \BB}):= \big\{\fv \in \oplus_{\ell = 1}^{n} L^{2}(\RE_{+}\!\times\RE) \;\big|\; \fv=\fv_{z} + G^{\oplus}_{z}\,\PP\,\big(\BB + M^{\oplus}_{z}\big)^{-1} \PP\,\gamma^{\oplus}_{+} \fv_{z}\,,\; \fv_{z}\!\in\!\dom\big(\LD^{\oplus}\big)\big\}
$$
according to
\begin{equation}\label{Lac}
\big(L_{\PP,\BB}-z\big)\fv = \big(\LD^{\oplus}-z\big)\fv_{z}\,.
\end{equation}
\end{theorem}

\begin{remark}\label{r:LK}
We use the notation $L_{K}$ (where $K$ stands for Kirchhoff) to denote the self-adjoint extension corresponding to the choices
\begin{equation}
\BB = 0\,,\qquad \PP 
=\frac1{n}\left(\begin{matrix}
1&1&\cdots&1\\
1&1&\cdots&1\\
\vdots&\vdots&\ddots&\vdots\\1&1&\cdots&1\end{matrix}
\right)\,. \label{KirBPi}
\end{equation}
We denote the associated resolvent operator with $R^{K}_z := (L_{K}-z)^{-1}$.
\end{remark}

In the sequel, we proceed to determine the unitary evolution associated to the above choices by means of arguments analogous to those described in the proof of \cite[Prop. 2.4]{CFP19}.

\begin{proposition} For all $\fv \in \oplus_{\ell=1}^{n} L^{2}(\RE_{+}\!\times\RE)$ and for all $t \in \RE$ there holds
\begin{equation}\label{semiclassdyn}
(e^{it L_{K}} \fv)(q,p) = \left(\!\!\begin{array}{c} \big(e^{it L_{0}} (0\!\oplus\! f_{1})\big)(q,p) \\ \vdots \\ \big(e^{it L_{0}} (0\!\oplus\! f_{n})\big)(q,p)\end{array}\!\!\right) 
- {\MC} \!\left(\!\!\begin{array}{c} \big(e^{it L_{0}} (0\!\oplus\! f_{1})\big)(-q,-p) \\ \vdots \\ \big(e^{it L_{0}} (0\!\oplus\! f_{n})\big)(-q,-p)\end{array}\!\!\right)
\end{equation}
where ${\MC} := \unit - 2\, {\PP}$ was already defined in relation with the quantum scattering operator, see Eq.\,\eqref{barefoot}. 
\end{proposition}

\begin{proof}
Throughout the whole proof we work component-wisely, denoting with $\ell \in \{1,...,n\}$ a fixed index.
Let us first remark that the resolvent \eqref{ResClass}, with $\BB,\PP$ as in Eq. \eqref{KirBPi}, acts on any element $\fv \in \oplus_{\ell=1}^{n} L^{2}(\RE_{+}\!\times\RE)$ according to
\begin{align*}
\big(R^{K}_{z} \fv\big)_{\ell}(q,p) 
& = (\RD_{z}f_\ell)(q,p) + \big(2\,g_z(q,p)\big) \sum_{j' = 1}^n {\hat \pi_{\ell j'} \over 2 \,m^{\infty}_{z}(p)} \sum_{j = 1}^n \hat \pi_{j' j}\, (\gamma_{+}\RD_{z} f_{j})(p) \\
& = (\RD_{z}f_\ell)(q,p) + {g_z(q,p) \over m^{\infty}_{z}(p)} \frac1n \sum_{j = 1}^n  (\gamma_{+}\RD_{z} f_j)(p) \,.
\end{align*}
From the above relation we derive the following, recalling the explicit expressions for $g_z(q,p)$ and $m^{\infty}_{z}(p)$ given in Eqs. \eqref{gzdef} and  \eqref{mzdef}, as well as Eq. \eqref{Gcexp} for $\gamma_{+}\RD_{z} f = \Gc^{+}_{z} f$:
\begin{align*}
\big(R^{K}_{z} \fv\big)_{\ell}(q,p) = (\RD_{z}f_\ell)(q,p) + \theta(p Im z)\,\sgn(Im z)\,{2i m \over n\,|p|} \sum_{j = 1}^n \int_{0}^{\infty}\! dq'\, e^{i m z (q+q')/p} \, f_j(q',-p)\,. 
\end{align*}

We now proceed to compute the unitary operator $e^{-i t L_{K}}$ ($t \in \RE$) by inverse Laplace transform, using the above representation for the resolvent $R^K_{z}$. Let us first assume $t > 0$; then, for any $c > 0$ and $\fv\!\in\! \dom(L_{K})$, we get (see \cite[Ch. III, Cor. 5.15]{EN})
$$ 
\big(e^{-it L_{K}}\fv\big)_{\ell} = {1 \over 2\pi i} \lim_{r \to \infty}\int_{- r + i c}^{r + i c}\! dz\; e^{-i z t}\,\big(R^{K}_{z} \fv\big)_{\ell} = {e^{c t} \over 2\pi i} \lim_{r \to \infty}\int_{-r}^{r}\! dk\; e^{- i t k}\,\big(R^{K}_{k + i c}\, \fv\big)_{\ell}\;.
$$ 
On the one hand, recalling Eq.\,\eqref{dynClaDir2} we have
\begin{align*}
 {e^{c t} \over 2\pi i} \lim_{r \to \infty}\int_{-r}^{r}\! dk\; e^{- i t k}\,(\RD_{k + i c}f_\ell)(q,p) 
& =  (e^{-it \LD} f_{\ell})(q,p) \\
& =  \big(e^{-i t L_{0}} (0\oplus f_{\ell})\big)(q,p) - \big(e^{-i t L_{0}} (0\oplus f_{\ell})\big)(-q,-p)\,. 
\end{align*}
On the other hand, noting that $\theta\big(p\Im (k\!+\!ic)\big) = \theta(p c) = \theta(p)$ and $\sgn\!\big(\Im (k\!+\!ic)\big) = \sgn(c) = +1$ for $c > 0$, by computations similar to those reported in the proof of \cite[Prop. 2.4]{CFP19} we get\footnote{Especially, recall the following basic identity regarding the unitary Fourier transform $\mathfrak{F}$ and its inverse $\mathfrak{F}^{-1}$:
$$
\mathfrak{F}^{-1}\Big(\mathfrak{F}h\big(a(\ast + ic)\big)\Big)(q) = {e^{c q} \over |a|}\,h(q/a) \qquad \mbox{for $a \in \RE \backslash \{0\}$, $c>0$, $q \in \RE$}\,,
$$
which holds true whenever ${e^{c \cdot} \over |a|}\,h(\cdot /a) \in L^2(\RE)$. In addition, keep in mind the relation written in Eq. \eqref{idfree1}.}
\begin{align*}
& {e^{c t} \over 2\pi i} \lim_{r \to \infty}\int_{-r}^{r}\!\! dk\; e^{- i t k}\,
\theta\big(p\Im (k\!+\!ic)\big) \sgn\!\big(\Im (k\!+\!ic)\big)\,{2 i m \over n |p|} \sum_{j = 1}^n \int_{0}^{\infty}\! dq'\, e^{i m (k+ic) (q+q')/p}\,f_{j}(q',-p) \\
& = \theta(p)\,{2 m \over n |p|}\, {e^{c (t - m q/p)} \over 2\pi} \sum_{j = 1}^n \lim_{r\to \infty}\int_{-r}^{r}\!\! dk\; e^{i (m q/p- t) k} \int_{\RE}\! dq'\, e^{i m (k + ic) q'/p}\; \theta(q')\,f_{j}(q',-p) \\
& = \theta(p)\,{2 m \over n |p|}\, e^{c (t - m q/p)} \sum_{j = 1}^n 
\mathfrak{F}^{-1}\Big(\mathfrak{F}\big(\theta(\,\cdot\,)\,f_{j}(\,\cdot\,,-p)\big)\big(-m (\ast + ic)/p\big)\Big)(m q/p- t) \\
& = \theta(p)\, {2 \over n}\, \sum_{j = 1}^n \theta\!\left(\!-\, q \!+\! {p\,t \over m}\right) f_{j}\!\left(- q + {p\,t \over m},-p\right) \\
& = {2 \over n}\, \sum_{j = 1}^n \big(e^{- i t L_0}(0 \oplus f_j)\big)(-q,-p)\,.
\end{align*}
Summing up, the above relations imply
\begin{equation}
\big(e^{-it L_{K}}\fv\big)_{\ell} = \big(e^{-it L_{0}} (0 \oplus f_{\ell})\big)(q,p) - \big(e^{-it L_{0}} (0 \oplus f_{\ell})\big)(-q,-p) + {2 \over n} \sum_{j = 1}^n\! \big(e^{- i t L_0}(0 \oplus f_j)\big)(-q,-p)\,. \label{eitLn}
\end{equation}
 
For $t < 0$ one can perform similar computations, starting from the following identity where $c > 0$:
$$ 
\big(e^{-it L_{K}}\fv\big)_{\ell} = -\,{1 \over 2\pi i} \lim_{r \to \infty}\int_{- r -i c}^{r -i c}\! dz\; e^{-i z t}\,\big(R^{K}_{z} \fv\big)_{\ell} = -\,{e^{-c t} \over 2\pi i} \lim_{r \to \infty}\int_{-r}^{r}\! dk\; e^{-i t k}\,\big(R^{K}_{k - i c}\, \fv\big)_{\ell}\;.
$$ 
We omit the related details for brevity. In the end, one obtains exactly the same expression as in Eq.\,\eqref{eitLn}, which with the trivial replacement $t \to -t$ proves Eq. \eqref{semiclassdyn}.
\end{proof}

\begin{remark}By Eq. \eqref{idfree1} we infer that the action of the unitary group $e^{it L_{K}}$ is explicitly given, component-wisely,  by Eq. \eqref{expLKgroup}. 
\end{remark}

\begin{remark} Recalling the explicit form of $\PP$ (see Eq. \eqref{KirBPi}) we obtain
\begin{equation}\label{MCexp}
\MC = \frac{1}{n}\!\left(\begin{matrix}
n - 2	&	-2		&	\cdots	&	-2	\\
-2		&	n - 2	&	\cdots	&	-2	\\
\vdots	&	\vdots	&	\ddots	&	\vdots\\
-2		&	-2		&	\cdots	&	n - 2
\end{matrix}\right), \quad
1 - \PP =\frac{1}{n}\!\left(\begin{matrix}
n - 1	&	-1		&	\cdots	&	-1	\\
-	1	&	n - 1	&	\cdots	&	-1	\\
\vdots	&	\vdots	&	\ddots	&\vdots	\\
-1		&	-1		&	\cdots	&	n-1
\end{matrix}\right). 
\end{equation}
In particular, for a star-graph with three edges ($n = 3$) we have
$$ 
\MC = {1 \over 3}\!\left(\begin{matrix}
1	&	-2	&	-2	\\
-2	&	1	&	-2	\\
-2	&	-2	&	1
\end{matrix}\right),
$$
whence $\MC = -\,\mathbb{M}$ with respect to the notation used in \cite{ACFN11}.
\end{remark}

\subsection{The semiclassical wave operators and scattering operator.}
Let us consider the wave operators and the corresponding scattering operator
on $\oplus_{\ell=1}^{n} L^{2}(\RE_{+}\!\times\RE)$, respectively defined by Eqs. \eqref{Wcldef} and \eqref{Scldef}. The following proposition provides explicit expressions for these operators.

\begin{proposition}\label{prop:WScl} The limits in Eq. \eqref{Wcldef} exist point-wisely for any $\xi\equiv(q,p)\!\in\! \RE_{+}\!\times (\RE \backslash \{0\})$ and in $L^{2}(\RE_{+}\!\times\RE)$ for any $\fv \in \oplus_{\ell=1}^{n} L^{2}(\RE_{+}\!\times\RE)$; moreover, there holds
\begin{equation}\label{Wclexp}
\big(\OmC^{\pm}\, \fv\big)(q,p)
= \big[\unit - \,\theta(\mp p)\; 2\,\PP\big]\, \fv(q,p)
= \big[\theta(\pm p)\, \unit + \,\theta(\mp p)\,\MC\big]\, \fv(q,p)\,.
\end{equation}
Furthermore, the scattering operator is given by
\begin{equation}\label{Sclexp}
\SC = \unit - 2\,{\PP} = {\MC}\,.
\end{equation}
\end{proposition}

\begin{proof} 
First of all let us point out that, for all $t \in \RE$ and any $\ell \in \{1,...,n\}$, Eqs. \eqref{dynClaDir2} and \eqref{semiclassdyn} give, respectively,
\begin{align*}
\big(e^{i t \LD} f\big)(q,p) 
& = \big(e^{i t L_{0}} (0\oplus f)\big)(q,p) - \big(e^{i t L_{0}} (0\oplus f)\big)(-q,-p) \\
& = (0\oplus f)\!\left(q + {pt \over m},p\right) - (0\oplus f)\!\left(-q - {pt \over m},-p\right) \\
& = \theta\!\left(q + {pt \over m}\right) f\!\left(q + {pt \over m},p\right) - \theta\!\left(-q - {pt \over m}\right) f\!\left(-q - {pt \over m},-p\right);
\end{align*}
\begin{align*}
(e^{-it L_{K}} \fv)_{\ell}(q,p) 
& = \big(e^{-it L_{0}} (0 \oplus f_{\ell})\big)(q,p) + \sum_{j = 1}^n\! \big(2 {\PP} \!-\! \unit\big)_{\ell j}\, \big(e^{-it L_{0}} (0 \oplus  f_{j})\big)(-q,-p) \\
& = (0 \oplus f_{\ell})\!\left(q - {p t \over m},p\right) + \sum_{j = 1}^n\! \big(2 {\PP} \!-\! \unit\big)_{\ell j}\, (0\oplus  f_{j})\!\left(-q + {pt \over m},-p\right) \\
& = \theta\!\left(q - {p t \over m}\right) f_{\ell}\!\left(q - {p t \over m},p\right) + \sum_{j = 1}^n\! \big(2 {\PP} \!-\! \unit\big)_{\ell j}\, \theta\!\left(-q + {pt \over m}\right) f_{j}\!\left(-q + {pt \over m},-p\right) .
\end{align*}
In view of the above relations, recalling that $q > 0$, by direct computations we obtain
\begin{align*}
& \big(e^{i t \LD^{\oplus}}e^{- i t L_{K}}\, \fv\big)_{\ell}(q,p) = \big(e^{i t \LD} (e^{- i t L_{K}}\, \fv)_{\ell}\big)(q,p) \\
& = \theta\!\left(q + {pt \over m}\right) (e^{- i t L_{K}}\, \fv)_{\ell}\!\left(q + {pt \over m}\,,p\right) - \theta\!\left(-q - {pt \over m}\right) (e^{- i t L_{K}}\, \fv)_{\ell}\!\left(-q - {pt \over m}\,,-p\right) \\
& = \theta\!\left(q + {pt \over m}\right) f_{\ell}(q,p) - \theta\!\left(-q - {pt \over m}\right) \sum_{j = 1}^n\! \big(2 {\PP} \!-\! \unit\big)_{\ell j}\, f_{j}(q,p) \\
& = f_{\ell}(q,p) - \theta\!\left(-q - {pt \over m}\right) \sum_{j = 1}^n\! 2\,\pi_{\ell j}\, f_{j}(q,p)\,,
\end{align*}
which gives
\begin{align*}
\big(e^{i t \LD^{\oplus}}e^{- i t L_{K}}\, \fv\big)(q,p) = \fv(q,p) - \theta\!\left(-q - {pt \over m}\right) \big(2\,{\PP}\, \fv\big)(q,p)\,.
\end{align*}
Then Eq. \eqref{Wclexp} follows noting that $\theta\!\left(-q - {pt \over m}\right) \to \theta(\mp p)$ for $t \to \pm\infty$.

Next, since $\PP$ is symmetric, by elementary arguments we get
\begin{equation*}
\big((\OmC^{\pm})^{*}\, \fv\big)(q,p) = \big[\unit - \,\theta(\mp p)\; 2\,\PP\big]\, \fv(q,p)\,.
\end{equation*}
On account of the above identity we obtain
\begin{align*}
\big(\SC\, \fv\big)(q,p) & = \big((\OmC^{+})^{*}\, \OmC^{-}\, \fv\big)(q,p) = \big[\unit - \,\theta(- p)\; 2\,\PP\big]\, \big(\OmC^{-}\, \fv\big)(q,p) \\
& = \big[\unit - \,\theta(- p)\; 2\,\PP\big]\, \big[\unit - \,\theta(p)\; 2\,\PP\big]\, \fv(q,p)
= \big[\unit - 2\,\PP\big]\, \fv(q,p)\,,
\end{align*}
which proves Eq. \eqref{Sclexp}.
\end{proof}

\section{Comparison of the semiclassical and quantum theories\label{s:4}}
Recall the definition of coherent states in $L^2(\RE)$ given in Eq. \eqref{initial}, and Eq. \eqref{exact} describing their free evolution. 

For later reference let us point out the following auxiliary result, regarding the functions $\widetilde{\psi}^{\hbar}_{\sigma,\xi}$ defined in Eq. \eqref{psitil}, and the operators $\mathcal{U}^{\pm}_{t}$ defined in Eq. \eqref{UUpm}.

\begin{lemma}\label{lem:Upm} For all $\xi = (q,p)\in \RE_{+}\!\times \RE$ and for any $t \in \RE$ there holds
$$
\big\|\,\mathcal{U}^{\pm}_{t}\, \widetilde{\psi}^{\hbar}_{\sigma_0,\xi} - e^{\frac{i}{\hbar}\Szero_{t}}\,\widetilde{\psi}^{\hbar}_{\sigma_t,\mp\xi_t}\big\|_{L^2(\RE_{+})} \leq {1 \over \sqrt{2}}\,e^{-\frac{q^{2}}{4\hbar \sigma_0^2}}\,.
$$
\end{lemma}

\begin{proof} Let us first remark that, on account of the considerations reported in Remark \ref{rem:RUpm}, $ \widetilde{\psi}^{\hbar}_{\sigma_0,\xi} = \Theta^* {\psi}^{\hbar}_{\sigma_0,\xi}$ and  for any $x > 0$ we have the chain of identities
\begin{align*}
\big(\mathcal{U}^{\pm}_{t}\, \widetilde{\psi}^{\hbar}_{\sigma_0,\xi}\big)(x) 
& = \big(\Ufree_{t}\, \theta\, \psi^{\hbar}_{\sigma_0,\xi}\big)(\mp x)
= \big(\Ufree_{t}\, \psi^{\hbar}_{\sigma_0,\xi}\big)(\mp x) + E^\hbar_t(\mp x) 
= e^{\frac{i}{\hbar}\Szero_{t}}\,\psi^{\hbar}_{\sigma_t,\xi_t}(\mp x) + E^\hbar_t(\mp x)\\ 
& = e^{\frac{i}{\hbar}\Szero_{t}}\,\psi^{\hbar}_{\sigma_t,\mp \xi_t}(x) + E^\hbar_t(\mp x)
= e^{\frac{i}{\hbar}\Szero_{t}}\,\widetilde{\psi}^{\hbar}_{\sigma_t,\mp \xi_t}(x) + E^\hbar_t(\mp x)\,,
\end{align*}
where we put $E^\hbar_t(x) := - \,\big(\Ufree_{t}\,(1-\theta)\,\psi^{\hbar}_{\sigma_0,\xi}\big)(x)$ for convenience of notation and used the trivial identity ${\psi}^{\hbar}_{\sigma,\xi}(- x) ={\psi}^{\hbar}_{\sigma, - \xi}(x) $.\\
On the other hand, via an explicit calculation involving the trivial inequality $e^{-\eta (a+b)^2} \leq e^{-\eta (a^2 + b^2)}$ for $\eta, a, b \geq 0$ we obtain
\begin{align}
\big\|E^\hbar_t (\mp \,\cdot\, )\big\|^{2}_{L^2(\RE_+)} 
& \leq \big\|E^\hbar_t\big\|^{2}_{L^2(\RE)} = \big\|\Ufree_{t}\, (1-\theta)\, \psi^{\hbar}_{\sigma_0,\xi}\big\|^{2}_{L^2(\RE)} = \big\| (1-\theta)\, \psi^{\hbar}_{\sigma_0,\xi}\big\|^{2}_{L^2(\RE)}  \nonumber \\
& = \int_{-\infty}^0\!\! dx \; |\psi^{\hbar}_{\sigma_0,\xi}(x)|^2  
 = \frac{1}{\sqrt{2\pi\hbar}\; \sigma_0}\int_0^\infty\! dx\; e^{-\frac{(x+q)^{2}}{2\hbar \sigma_0^2}} 
\leq {1 \over 2}\, e^{-\frac{q^2}{2\hbar \sigma_0^2}},
\label{gallery}
\end{align}
which proves the thesis in view of the previous arguments.
\end{proof}

To proceed let us point out the  forthcoming lemma which  characterizes a large class of functions satisfying the condition in Definition \ref{d:quantumstates}.

\begin{lemma}\label{lem:COS} Let $\eta > 0$, $\xi = (q,p) \in \RE_{+}\!\times \RE$ and consider a family of functions $\chi_{q,\eta} \in L^{\infty}(\RE_{+})$, uniformly bounded in $L^{\infty}(\RE_{+})$ with respect to $\eta,q$ and such that 
\begin{equation}\label{hypchi}
\chi_{q,\eta}(x) = 1 \qquad \mbox{for\; $|x-q| < \eta\, q$}\,.
\end{equation}
Then, the functions $\Xi^{\hbar}_{\sigma,\xi} \in L^2(\RE_{+})$ defined by
\begin{equation}
\Xi^{\hbar}_{\sigma,\xi}(x) := {1 \over \|\chi_{q,\eta}\,\widetilde{\psi}^{\hbar}_{\sigma,\xi}\|_{L^2(\RE_{+})}}\;\,\chi_{q,\eta}(x)\, \widetilde{\psi}^{\hbar}_{\sigma,\xi}(x) \qquad (x > 0) \label{COSlem}
\end{equation}
fulfill the condition \eqref{COS} with $\eps < \min\{1/4,\eta^2/8\}$ and   $C_0$ depending only on $\|\chi_{q,\eta}\|_{L^{\infty}(\RE_{+})}$.
\end{lemma}

\begin{proof} Let us first remark that the states $\Xi^{\hbar}_{\sigma,\xi}$ defined in Eq. \eqref{COSlem} have unit norm in $L^2(\RE_+)$ by construction; moreover, again from Eq. \eqref{COSlem} it follows that $\chi_{q,\eta}\, \widetilde{\psi}^{\hbar}_{\sigma,\xi} = \|\chi_{q,\eta}\,\widetilde{\psi}^{\hbar}_{\sigma,\xi}\|_{L^2(\RE_{+})}\, \Xi^{\hbar}_{\sigma,\xi}$. Taking these facts into account, we have
\begin{align*}
\big\|\Xi^{\hbar}_{\sigma,\xi} - \widetilde{\psi}^{\hbar}_{\sigma,\xi}\big\|_{L^2(\RE_{+})} 
& \leq \big\|\Xi^{\hbar}_{\sigma,\xi} - \chi_{q,\eta}\,\widetilde{\psi}^{\hbar}_{\sigma,\xi}\big\|_{L^2(\RE_{+})} + \big\|(\chi_{q,\eta}-1)\,\widetilde{\psi}^{\hbar}_{\sigma,\xi}\big\|_{L^2(\RE_{+})} \\
& = \big|1 - \|\chi_{q,\eta}\,\widetilde{\psi}^{\hbar}_{\sigma,\xi}\|_{L^2(\RE_{+})}\big|\; \big\|\Xi^{\hbar}_{\sigma,\xi}\big\|_{L^2(\RE_{+})} + \big\|(\chi_{q,\eta} - 1)\,\widetilde{\psi}^{\hbar}_{\sigma,\xi}\big\|_{L^2(\RE_{+})} \\
& = \big|1 - \|\chi_{q,\eta}\,\widetilde{\psi}^{\hbar}_{\sigma,\xi}\|_{L^2(\RE_{+})}\big| + \big\|(\chi_{q,\eta}-1)\,\widetilde{\psi}^{\hbar}_{\sigma,\xi}\big\|_{L^2(\RE_{+})}\,.
\end{align*}

On one hand, recalling the definition \eqref{psitil} of $\widetilde{\psi}^{\hbar}_{\sigma,\xi}$ and that $\|\psi^{\hbar}_{\sigma,\xi}\|_{L^2(\RE)} = 1$, using the basic inequality $(a-b)^2 \leq |a^2 - b^2|$ for $a,b> 0$ we get
\begin{align*}
& \big|1 - \|\chi_{q,\eta}\,\widetilde{\psi}^{\hbar}_{\sigma,\xi}\|_{L^2(\RE_{+})}\big|^2 
\leq \big|\|\psi^{\hbar}_{\sigma,\xi}\|^2_{L^2(\RE)} - \|\chi_{q,\eta}\,\widetilde{\psi}^{\hbar}_{\sigma,\xi}\|_{L^2(\RE_{+})}^2 \big| \\
& = \left|\int_{\RE}\! dx\; |\psi^{\hbar}_{\sigma,\xi}(x)|^2 - \int_{0}^{\infty}\!\! dx\; \big|\chi_{q,\eta}(x)\,\psi^{\hbar}_{\sigma,\xi}(x)\big|^2 \right| 
\leq \int_{\RE}\!\! dx \left|1- \theta(x)\, \big|\chi_{q,\eta}(x)\big|^2\right|\,\big|\psi^{\hbar}_{\sigma,\xi}(x)\big|^2 \\
& \leq \int_{\RE}\!\! dx\; \big(1- \theta(x)\big)\;\big|\psi^{\hbar}_{\sigma,\xi}(x)\big|^2 + \int_{\RE}\!\! dx\;\theta(x) \left|1- \big|\chi_{q,\eta}(x)\big|^2\right|\;\big|\psi^{\hbar}_{\sigma,\xi}(x)\big|^2\,.
\end{align*}
Recalling the hypothesis \eqref{hypchi}, the explicit expression \eqref{initial} for $\psi^{\hbar}_{\sigma,\xi}$,  and that we are assuming $q > 0$, from the above results we derive
\begin{align*}
& \big|1 - \|\chi_{q,\eta}\,\widetilde{\psi}^{\hbar}_{\sigma,\xi}\|_{L^2(\RE_{+})}\big|^2 
\leq \int_{-\infty}^{0}\!\! dx\; \big|\psi^{\hbar}_{\sigma,\xi}(x)\big|^2 + \int_{\RE_{+}\, \cap\, \{|x-q| \,>\, \eta\,q\}}\!\!\!\! dx \left|1- \big|\chi_{q,\eta}(x)\big|^2\right|\;\big|\psi^{\hbar}_{\sigma,\xi}(x)\big|^2 \\
& \leq \frac{1}{\sqrt{2\pi\hbar}\;|\sigma|} \left[
\int_{0}^{\infty}\!\! dx\; e^{-\,\frac{(x+q)^{2}}{2\hbar |\sigma|^2}}\,
+ \big(1 + \|\chi_{q,\eta}\|_{L^{\infty}(\RE_{+})}^2\big)
\int_{\{|x-q| \,>\, \eta\,q\}}\!\!\!\! dx\; e^{-\,\frac{(x-q)^{2}}{2\hbar |\sigma|^2}}\right] \\
& \leq \frac{1}{\sqrt{2\pi\hbar}\;|\sigma|} \left[
e^{-\,\frac{q^2}{2\hbar |\sigma|^2}} \int_{0}^{\infty}\!\! dx\; e^{-\,\frac{x^2}{2\hbar |\sigma|^2}}\,
+ \big(1 + \|\chi_{q,\eta}\|_{L^{\infty}(\RE_{+})}\big)^2\,
e^{-\,\frac{\eta^2 q^2}{4\hbar |\sigma|^2}} \int_{\RE}\!\! dx\;  e^{-\,\frac{(x-q)^{2}}{4\hbar |\sigma|^2}}\right] \\
& \leq {1 \over 2}\;e^{-\,\frac{q^2}{2\hbar |\sigma|^2}} 
+ \sqrt{2}\;\big(1 + \|\chi_{q,\eta}\|_{L^{\infty}(\RE_{+})}\big)^2\; e^{-\,\frac{\eta^2 q^2}{4\hbar |\sigma|^2}}\,.
\end{align*}

On the other hand, by arguments similar to those employed above we get
\begin{align*}
& \big\|(\chi_{q,\eta}-1)\,\widetilde{\psi}^{\hbar}_{\sigma,\xi}\big\|_{L^2(\RE_{+})}^2 
= \int_{0}^{\infty}\!\!dx\;\big|\chi_{q,\eta}(x)-1\big|^2\;\big|\psi^{\hbar}_{\sigma,\xi}(x)\big|^2 \\
& \leq \big(1+\|\chi_{q,\eta}\|_{L^{\infty}(\RE_{+})}\big)^2\; \frac{1}{\sqrt{2\pi\hbar}\;|\sigma|} \int_{\{|x-q| \,>\, \eta\,q\}}\!\!dx\;e^{-\,\frac{(x-q)^{2}}{2\hbar |\sigma|^2}} 
\leq \sqrt{2}\;\big(1+\|\chi_{q,\eta}\|_{L^{\infty}(\RE_{+})}\big)^2\, e^{-\,\frac{\eta^2 q^2}{4\hbar |\sigma|^2}}\,.
\end{align*}

Summing up, the previous results and the basic relation $\sqrt{a+b} \leq \sqrt{a} + \sqrt{b}$ for $a,b > 0$ imply
\begin{align*}
\big\|\Xi^{\hbar}_{\sigma,\xi} - \widetilde{\psi}^{\hbar}_{\sigma,\xi}\big\|_{L^2(\RE_{+})} 
& \leq {1 \over \sqrt{2}} \;e^{-\,\frac{q^2}{4\hbar |\sigma|^2}} + 2^{5/4} \big(1 + \|\chi_{q,\eta}\|_{L^{\infty}(\RE_{+})}\big)\; e^{-\,\frac{\eta^2 q^2}{8\hbar |\sigma|^2}}\\
& \leq \left({1 \over \sqrt{2}} + 2^{5/4} \big(1 + \|\chi_{q,\eta}\|_{L^{\infty}(\RE_{+})}\big)\right) e^{-\,\eps\frac{ q^2}{\hbar |\sigma|^2}}  \,,
\end{align*}
which suffices to infer the thesis on account of the uniform boundedness of $\chi_{q,\eta}$.
\end{proof}

\begin{example} For $\eta \in (0,1]$, consider the sharp cut-off functions
$$ \chi_{q,\eta}(x) = \left\{\!\begin{array}{ll}
\dd{0}	&	\quad\dd{\mbox{if\, $x \leq (1-\eta)\,q$}}\,, \vspace{0.1cm}\\
\dd{1}	&	\quad\dd{\mbox{if\, $x > (1-\eta)\,q$}}\,,
\end{array}\right. $$
which clearly satisfy the hypothesis of Lemma \ref{lem:COS}. The corresponding elements $\Xi^{\hbar}_{\sigma,\xi} \in L^2(\RE_{+})$ defined according to Eq. \eqref{COSlem} consist of normalized truncations of the coherent state $\psi^{\hbar}_{\sigma,\xi}$ and fulfill the condition \eqref{COS} as a consequence.\\
It is worth noting that for $\eta = 1$ we have $\chi_{q,\eta} \equiv 1$ on $\RE_{+}$, so that the associated function $\Xi^{\hbar}_{\sigma,\xi} $ is just the re-normalization of the bare truncation $\widetilde{\psi}^{\hbar}_{\sigma,\xi}$ introduced in Eq. \eqref{psitil}, i.e.,
$$
\Xi^{\hbar}_{\sigma,\xi} = \widetilde{\psi}^{\hbar}_{\sigma,\xi}\,\big/\,\|\widetilde{\psi}^{\hbar}_{\sigma,\xi}\|_{L^2(\RE_{+})}\;.
$$
\end{example}

\begin{example} For $\eta \in (0,1/2)$, consider the smooth functions on $\RE_{+}$ such that
$$ \chi_{q,\eta}(x) = \left\{\!\begin{array}{ll}
\dd{0}	&	\quad\dd{\mbox{for\, $|x - q| > (1-\eta)\,q$}\,,} \vspace{0.1cm}\\
\dd{1}	&	\quad\dd{\mbox{for\, $|x - q| < \eta\,q$}\,,}
\end{array}\right. \qquad
\big|\chi_{q,\eta}(x)\big| \leq 1\,.
$$
Again, the assumptions of Lemma \ref{lem:COS} are certainly verified and the related functions $\Xi^{\hbar}_{\sigma,\xi}$
have compact support in $\RE_{+}$, besides satisfying the bound \eqref{COS}.
\end{example}

In addition to states fulfilling the requirements of Definitions \ref{d:quantumstates} and \ref{d:classicalstates}, our arguments will often involve the non-normalized element
\begin{equation}
\bm{\Psi}^{\hbar}_{\sigma_0,\xi} \in L^{2}(\GG)\,, \qquad \bm{\Psi}^{\hbar}_{\sigma_0,\xi} \equiv \left(\!\!\begin{array}{c} \widetilde{\psi}^{\hbar}_{\sigma_0,\xi} \vspace{0.05cm}\\ 0 \vspace{-0.15cm}\\ \vdots \\ 0 \end{array}\!\!\right) , \label{PsiGraph}
\end{equation}
along with its classical counterpart
\begin{equation}
\bm{\Phi}^{\hbar}_{\sigma_0,x} \in L^{2}(\GG)\,, \qquad \bm{\Phi}^{\hbar}_{\sigma_0,x} \equiv \left(\!\!\begin{array}{c} \widetilde{\phi}^{\hbar}_{\sigma_0,x} \vspace{0.05cm}\\ 0 \vspace{-0.15cm}\\ \vdots \\ 0 \end{array}\!\!\right), \label{PsiGraphCl}
\end{equation}
with $\widetilde{\psi}^{\hbar}_{\sigma_0,\xi}$ defined as in Eq. \eqref{psitil} and
\begin{equation*}\label{PsiCl}
\widetilde{\phi}^{\hbar}_{\sigma_0,x}(\xi) := \widetilde{\psi}^{\hbar}_{\sigma_0,\xi}(x)\,.
\end{equation*}

\subsection{Comparing the dynamics. Proof of Theorem \ref{t:dynamics}} 
\begin{proof}[{\bf Proof of Theorem \ref{t:dynamics}}] Let $\bm{\Psi}^{\hbar}_{\sigma_0,\xi}$ and $\bm{\Phi}^{\hbar}_{\sigma_0,(\cdot)}(\xi)$ be, respectively, as in Eqs. \eqref{PsiGraph} and \eqref{PsiGraphCl}, and note that from the triangular inequality it follows
\begin{align}
& \big\| e^{-i {t \over \hbar} \HK} \bm{\Xi}^{\hbar}_{\sigma_0,\xi} - e^{\frac{i}{\hbar}\Szero_{t}} \big( e^{it L_{K}} \bm{\Sigma}^{\hbar}_{\sigma_t,(\cdot)}\big)(\xi)\big\|_{L^2(\GG)} \nonumber\\
& \leq \big\| e^{-i {t \over \hbar} \HK} \bm{\Xi}^{\hbar}_{\sigma_0,\xi} - e^{-i {t \over \hbar} \HK} \bm{\Psi}^{\hbar}_{\sigma_0,\xi}\big\|_{L^2(\GG)} 
+ \big\|e^{-i {t \over \hbar} \HK} \bm{\Psi}^{\hbar}_{\sigma_0,\xi} - e^{\frac{i}{\hbar}\Szero_{t}} \big( e^{it L_{K}} \bm{\Phi}^{\hbar}_{\sigma_t,(\cdot)}\big)(\xi)\big\|_{L^2(\GG)} \nonumber\\
& \quad + \big\| e^{\frac{i}{\hbar}\Szero_{t}} \big( e^{it L_{K}} \bm{\Phi}^{\hbar}_{\sigma_t,(\cdot)}\big)(\xi) - e^{\frac{i}{\hbar}\Szero_{t}} \big( e^{it L_{K}} \bm{\Sigma}^{\hbar}_{\sigma_t,(\cdot)}\big)(\xi)\big\|_{L^2(\GG)}\,. \label{proof1}
\end{align}
Regarding the first term on the right-hand side of Eq. \eqref{proof1}, by the unitarity of $e^{-i {t \over \hbar} \HK}$ and the condition \eqref{COS} we infer
\begin{align*}
\big\| e^{-i {t \over \hbar} \HK} \bm{\Xi}^{\hbar}_{\sigma_0,\xi} - e^{-i {t \over \hbar} \HK} \bm{\Psi}^{\hbar}_{\sigma_0,\xi}\big\|_{L^2(\GG)} 
= \big\|\bm{\Xi}^{\hbar}_{\sigma_0,\xi} - \bm{\Psi}^{\hbar}_{\sigma_0,\xi}\big\|_{L^2(\GG)} 
= \big\|\Xi^{\hbar}_{\sigma_0,\xi} - \widetilde{\psi}^{\hbar}_{\sigma_0,\xi}\big\|_{L^2(\RE_{+})} 
\leq C_0 e^{- \,\eps\,{q^2 \over \hbar \sigma_0^2}}\,.
\end{align*}
As for the second term in Eq. \eqref{proof1}, note that Eqs. \eqref{dynHF} and \eqref{semiclassdyn}  give
\begin{align*}
& e^{-i {t \over \hbar} \HK} \bm{\Psi}^{\hbar}_{\sigma_0,\xi} - e^{\frac{i}{\hbar}\Szero_{t}} \big( e^{it L_{K}} \bm{\Phi}^{\hbar}_{\sigma_t,(\cdot)}\big)(\xi) \\
& = 
\left(\!\!\begin{array}{c} \mathcal{U}^{-}_{t}\, \widetilde{\psi}^{\hbar}_{\sigma_0,\xi} - e^{\frac{i}{\hbar}\Szero_{t}} \big(e^{it L_{0}} (0\!\oplus\! \widetilde{\phi}^{\hbar}_{\sigma_t,(\cdot)})\big)(\xi) \vspace{0.05cm}\\ 0 \vspace{-0.15cm}\\ \vdots \\ 0 \end{array}\!\!\right) \!
- {\MC} \left(\!\!\begin{array}{c} \mathcal{U}^{+}_{t}\, \widetilde{\psi}^{\hbar}_{\sigma_0,\xi} - e^{\frac{i}{\hbar}\Szero_{t}} \big(e^{it L_{0}} (0\!\oplus\! \widetilde{\phi}^{\hbar}_{\sigma_t,(\cdot)})\big)(-\xi) \vspace{0.05cm}\\ 0 \vspace{-0.15cm}\\ \vdots \\ 0 \end{array}\!\!\right).
\end{align*}
Since $\big(e^{it L_{0}} (0\!\oplus\! \widetilde{\phi}^{\hbar}_{\sigma_t,x})\big)(\pm \xi) = \widetilde{\phi}^{\hbar}_{\sigma_t,x}(\pm \xi_t) = \widetilde{\psi}^{\hbar}_{\sigma_t,\pm \xi_t}(x)$ for $x \in \RE_{+}$, from the above identity and from Lemma \ref{lem:Upm} we deduce\footnote{Note also that, on account of Eq. \eqref{MCexp}, we have $$1 + \sum_{\ell = 1}^{n} |\MC_{\ell 1}|^2 = 1 + |\MC_{1 1}|^2 + \sum_{\ell = 2}^{n} |\MC_{\ell 1}|^2 = 1 + \left({n-2 \over n}\right)^{\!2} + (n-1)\,\left({2 \over n}\right)^{\!\!2} = 2\,. $$}
\begin{align*}
& \big\|e^{-i {t \over \hbar} \HK} \bm{\Psi}^{\hbar}_{\sigma_0,\xi} - e^{\frac{i}{\hbar}\Szero_{t}} \big( e^{it L_{K}} \bm{\Phi}^{\hbar}_{\sigma_t,(\cdot)}\big)(\xi)\big\|_{L^2(\GG)}^2 \\
& \leq 2\, \big\|\mathcal{U}^{-}_{t}\, \widetilde{\psi}^{\hbar}_{\sigma_0,\xi} - e^{\frac{i}{\hbar}\Szero_{t}} 
\widetilde{\psi}^{\hbar}_{\sigma_t, \xi_t}\big\|_{L^2(\RE_{+})}^2
+ 2 \sum_{\ell = 1}^{n} |\MC_{\ell 1}|^2\, \big\|\mathcal{U}^{+}_{t}\, \widetilde{\psi}^{\hbar}_{\sigma_0,\xi} - e^{\frac{i}{\hbar}\Szero_{t}} \widetilde{\psi}^{\hbar}_{\sigma_t,- \xi_t}\big\|_{L^2(\RE_{+})}^2
\leq 2 \, e^{-\frac{q^{2}}{2\hbar \sigma_{0}^{2}}}\,.
\end{align*}
Let us finally consider the third term in Eq. \eqref{proof1}. Recalling again the identity \eqref{semiclassdyn}, we obtain
\begin{align*}
& \big( e^{it L_{K}} \bm{\Phi}^{\hbar}_{\sigma_t,(\cdot)}\big)(\xi) - \big( e^{it L_{K}} \bm{\Sigma}^{\hbar}_{\sigma_t,(\cdot)}\big)(\xi) \\
& = \left(\!\!\begin{array}{c} \big(e^{it L_{0}} (0\!\oplus\! \widetilde{\phi}^{\hbar}_{\sigma_t,(\cdot)})\big)(\xi) - \big(e^{it L_{0}} (0\!\oplus\! \widetilde{\Sigma}^{\hbar}_{\sigma_t,(\cdot)})\big)(\xi) \\ 0 \\\vdots \\ 0 \end{array}\!\!\right) \!
- {\MC} \!\left(\!\!\begin{array}{c} \big(e^{it L_{0}} (0\!\oplus\! \widetilde{\phi}^{\hbar}_{\sigma_t,(\cdot)})\big)(-\xi) - \big(e^{it L_{0}} (0\!\oplus\! \widetilde{\Sigma}^{\hbar}_{\sigma_t,(\cdot)})\big)(-\xi) \\ 0 \\ \vdots \\ 0 \end{array}\!\!\right).
\end{align*}
From the above identity, by arguments similar to those employed previously we get
\begin{align*}
& \big\|\big( e^{it L_{K}} \bm{\Phi}^{\hbar}_{\sigma_t,(\cdot)}\big)(\xi) - \big( e^{it L_{K}} \bm{\Sigma}^{\hbar}_{\sigma_t,(\cdot)}\big)(\xi)\big\|_{L^2(\GG)}^2 \\
& \leq 2\,\big\|\widetilde{\psi}^{\hbar}_{\sigma_t,\xi_t} - \Xi^{\hbar}_{\sigma_t,\xi_t}\big\|_{L^2(\RE_{+})}^2
+ 2 \sum_{\ell = 1}^{n} |\MC_{\ell 1}|^2\, \big\|\widetilde{\psi}^{\hbar}_{\sigma_t,-\xi_t} - \Xi^{\hbar}_{\sigma_t,-\xi_t}\big\|_{L^2(\RE_{+})}^2 \leq 4C_0^2\,e^{- \,2\eps\,{q_t^2 \over \hbar |\sigma_t|^2}}\,.
\end{align*}
Summing up, the above bounds imply the thesis.
\end{proof}

\subsection{Comparing the wave and scattering operators. Proof of Theorem \ref{t:waveoperators}}

\begin{proof}[{\bf Proof of  Theorem \ref{t:waveoperators}}]Note the identity $\MC\,(1,0,\dots,0)^T = (1-2/n,-2/n,\dots,-2/n)^T $ and recall the expression of $\OmQ^{\pm}$ given  in Eq. \eqref{house}. Then, by simple computations we get ($\ell = 1,\dots,n$)
\[
( \OmQ^{\pm}\, \bm{\Xi}^{\hbar}_{\sigma_0,\xi})_\ell = \left\{\begin{aligned}
&\bigg(1-\frac1n \,(1 \mp \mathcal F_c^*\mathcal F_s)\bigg) \Xi^{\hbar}_{\sigma_0,\xi}&\quad \mbox{if\, $\ell = 1$}\,, \\
&-\frac1n\, (1 \mp \mathcal F_c^*\mathcal F_s)\, \Xi^{\hbar}_{\sigma_0,\xi}&\quad \mbox{if\, $\ell \neq  1$}\,.
\end{aligned}\right. 
\]
Additionally, recalling  the expression of $\OmC^{\pm}$ given  in Eq. \eqref{Wclexp},  we get
\[\begin{aligned}
( \OmC^{\pm}\, \bm{\Sigma}^{\hbar}_{\sigma_0,x})_\ell (\xi) 
& = \left\{\begin{aligned}
&\bigg(1-\theta(\mp p)\,\frac2n\bigg)\Sigma^\hbar_{\sigma_0,x}(\xi) &\quad \mbox{if\, $\ell = 1$}\,, \\
&-\theta(\mp p)\,\frac2n\,\Sigma^\hbar_{\sigma_0,x}(\xi) &\quad \mbox{if\, $\ell \neq 1$}\,.
\end{aligned}\right. 
\end{aligned}\]
In view of these results  together with the identity \eqref{CS1} we derive
\[\begin{aligned}
\big\|\OmQ^{\pm}\, \bm{\Xi}^{\hbar}_{\sigma_0,\xi} - (\OmC^{\pm}\bm{\Sigma}^{\hbar}_{\sigma_0,(\cdot)})(\xi)\big\|_{L^2(\GG)} 
& =  \frac{1}{\sqrt n}\,
\big\|\big(1 \mp \mathcal F_c^*\mathcal F_s - 2 \theta(\mp p) \big)\,\Xi^{\hbar}_{\sigma_0,\xi}\big\|_{L^2(\RE_{+})} \\ 
& =  \frac{1}{\sqrt n}\,
\big\|\big((1 - 2 \theta(\mp p) ) \mathcal F_c \mp \mathcal F_s \big)\,\Xi^{\hbar}_{\sigma_0,\xi}\big\|_{L^2(\RE_{+})} \,.
\end{aligned}
\]
By the bound $\|(1 - 2 \theta(\mp p) ) \mathcal F_c \mp \mathcal F_s\| \leq |1 - 2 \theta(\mp p) | \|\mathcal F_c \| + \|\mathcal F_s \| \leq 2 $ and by Eq. \eqref{COS}, we infer 
\[
\big\|\big((1 - 2 \theta(\mp p) ) \mathcal F_c \mp \mathcal F_s \big)\Xi^{\hbar}_{\sigma_0,\xi}\big\|_{L^2(\RE_{+})} \leq 2 C_0\,  e^{- \,\eps\,{q^2 \over \hbar \sigma_0^2}} + 
\big\|\big((1 - 2 \theta(\mp p) ) \mathcal F_c \mp \mathcal F_s \big)\widetilde{\psi}^{\hbar}_{\sigma_0,\xi} \big\|_{L^2(\RE_{+})}.
\]

In what follows we prove the following upper bound  which concludes the proof of the theorem
\begin{equation}\label{tillerman}
\big\|\big((1 - 2 \theta(\mp p) ) \mathcal F_c \mp \mathcal F_s \big)\widetilde{\psi}^{\hbar}_{\sigma_0,\xi} \big\|_{L^2(\RE_{+})} 
\leq \sqrt{2}\, \Big(e^{-\frac{\sigma_0^2p^2}{\hbar }}+e^{-\frac{q^2}{4\hbar \sigma_0^2}}\Big)\,. 
\end{equation}

We start with the identity 
\[
\big\|\big((1 - 2 \theta(\mp p) ) \mathcal F_c \mp \mathcal F_s \big)\widetilde{\psi}^{\hbar}_{\sigma_0,\xi} \big\|_{L^2(\RE_{+})}^2 
=
\frac2\pi \int_0^\infty\!\! dk \left| \int_0^\infty\!\! dx \, \bigg(\big(1 - 2 \theta(\mp p)\big) \cos (kx) \pm i \sin(kx) \bigg)   {\psi}^{\hbar}_{\sigma_0,\xi}(x) \right|^2 . 
\]
Considering separately the cases $p>0$ and $p<0$ for the two possible choices of the signs, it is easy to convince oneself that 
\[
\big\|\big((1 - 2 \theta(\mp p) ) \mathcal F_c \mp \mathcal F_s \big)\widetilde{\psi}^{\hbar}_{\sigma_0,\xi} \big\|_{L^2(\RE_{+})}^2 
=\frac2\pi\left\{
\begin{aligned}
& \int_0^\infty\!\! dk \left| \int_0^\infty\!\! dx \; e^{ikx}  \,{\psi}^{\hbar}_{\sigma_0,\xi}(x) \right|^2 &\quad \mbox{if\, $p>0$}\,, \\ 
& \int_0^\infty\!\! dk \left| \int_0^\infty\!\! dx \; e^{- ikx}  \,  {\psi}^{\hbar}_{\sigma_0,\xi}(x) \right|^2 & \quad \mbox{if\, $p<0$}\,.
\end{aligned}\right.
\]
Recall that the Fourier transform of ${\psi}^{\hbar}_{\sigma_0,\xi}$ is given by
\[
\big(\mathcal {F}{\psi}^{\hbar}_{\sigma_0,\xi}\big) (k )
 : = \frac1{\sqrt{2\pi}} \int_\RE\! dx \; e^{-ikx} \,  \psi^{\hbar}_{\sigma_0,\xi} (x)= \sqrt{\sigma_0}  \left({2 \hbar \over \pi}\right)^{\!\!1/4} e^{-{\hbar \sigma_0^2}(k- p/\hbar)^2 -  i k q}\,. 
 \]
Let us assume $p>0$, we have the chain of inequalities/identities  
\[\begin{aligned}
\left(\frac2\pi\int_0^\infty\!\! dk \left| \int_0^\infty\!\! dx\; e^{ikx}\,  {\psi}^{\hbar}_{\sigma_0,\xi}(x) \right|^2   \right)^{1/2}
& = 2\;\big\| \big(\mathcal F\theta {\psi}^{\hbar}_{\sigma_0,\xi}\big)(-\,\cdot)\big\|_{L^2(\RE_+)}\\ 
& \leq 2\,\big\|\big(\mathcal F{\psi}^{\hbar}_{\sigma_0,\xi}\big)(-\,\cdot)\big\|_{L^2(\RE_+)} + 2\,\big\| \big(\mathcal F(1-\theta) {\psi}^{\hbar}_{\sigma_0,\xi}\big)(-\,\cdot)\big\|_{L^2(\RE_+)}\,.
\end{aligned}\]
Reasoning like for the bound in Eq. \eqref{gallery}, we obtain 
\[
\big\|\big(\mathcal F{\psi}^{\hbar}_{\sigma_0,\xi}\big)(-\,\cdot)\big\|_{L^2(\RE_+)} = \left(\int_0^\infty\!\! dk \left| \big(\mathcal F{\psi}^{\hbar}_{\sigma_0,\xi}\big)(-k ) \right|^2 \right)^{1/2}  
\leq {1 \over \sqrt{2}}\; e^{-\frac{\sigma_0^2p^2}{\hbar }}
\]
and 
\[\begin{aligned}
\big\| \big(\mathcal F(1-\theta) {\psi}^{\hbar}_{\sigma_0,\xi}\big)(-\,\cdot)\big\|_{L^2(\RE_+)} 
& \leq \big\| \mathcal F(1-\theta) {\psi}^{\hbar}_{\sigma_0,\xi}\big\|_{L^2(\RE)} \\ 
& = \big\| (1-\theta) {\psi}^{\hbar}_{\sigma_0,\xi}\big\|_{L^2(\RE)} =\left( \int_{-\infty}^0\!\! dx\, \left|  {\psi}^{\hbar}_{\sigma_0,\xi} (x)\right|^2\right)^{1/2} 
\leq {1 \over \sqrt{2}}\; e^{-\frac{q^2}{4\hbar \sigma_0^2}},
\end{aligned}\]
which conclude the proof of the bound \eqref{tillerman} for $p>0$. The proof of the bound for $p<0$ is identical and we omit it.

Identity \eqref{th2so} follows immediately from Eqs. \eqref{whatis1} and \eqref{whatis2}.
\end{proof}

\section{Final remarks\label{s:5}}

\subsection{A comparison with the standard approach to the definition of a classical dynamics on the graph}

There is no  trajectory of  a classical particle which is the semiclassical limit of $e^{-i {t \over \hbar} \HK} \bm{\Xi}^{\hbar}_{\sigma_0,\xi}$. As a consequence, the semiclassical  dynamics is not described by the Hamilton equations. One way to overcome this difficulty is to  assign a probability to every possible path on the graph. Typically the probability of a certain path is postulated, and given in terms of  the square modulus of the   quantum transition (or stability) amplitudes (see, e.g.,  \cite[Sec. II.A]{BG-pre01_1} or \cite[Sec. 6.1]{BK}). For a star-graph the latter coincide with the elements of the (vertex) scattering matrix, defined for generic boundary conditions, e.g.,  in \cite[Thm. 2.1]{KS-jpa99} or \cite[Lem. 2.1.3]{BK}. For  Kirchhoff boundary conditions the elements of the  scattering matrix are   given by $\frac{2}{n} - \delta_{\ell,\ell'}$, $\ell,\ell' = 1, \dots,n$,
see \cite[Eq. (1)]{BG-pre01_1} (for the star-graph $C_{bb'}=1$), and \cite[Ex. 2.1.7, p. 41]{BK}. This is the approach used (for compact graphs) by Kottos and Smilansky in \cite{KS-prl97} and in several other works, see, e.g., \cite{BG-pre01_1}, the review \cite{GS-ap06}, and the monograph \cite{BK}. We have already noted that, up to a sign, the coefficients $\frac{2}{n} - \delta_{\ell,\ell'}$  coincide with the elements of the matrix $\MC$ identifying both the classical and quantum scattering operators. 

In our paper we followed a different train of thought. We wanted to recover the limiting classical dynamics on the star-graph starting from the trivial dynamics of classical particle on the half-line with elastic collision in the origin. To do so we made use of a  Kre\u{\i}n's formula  to find and classify singular perturbations of self-adjoint operators, see \cite{P01} and  \cite{P08}. We remark that, in a similar way, one can reconstruct the Hamiltonian $H_K$ starting from the free Hamiltonian of a quantum particle on the half-line with Dirichlet conditions in the origin. 

We defined the generator of the  trivial dynamics  on the half-line through Eq. \eqref{group}.  Note that if $f \in \dom(\LD)$, then $f_t = e^{-i t \LD} f$ satisfies the Liouville equation 
\[
i\pd{}{t} f_t = \LD f\,,
\]
but the action of the group can be extended in a natural way to any bounded function. 

Since the evolution is unitary in $L^2(\RE_+\times \RE)$, if $\|f\|_{L^{2}(\RE_{+}\times\RE)} =1$, we can interpret 
\begin{equation*}
\rho_t(q,p) := \big|\big(e^{-i t \LD}  f\big)(q,p)\big|^2
\end{equation*}
as a density in the phase space $\RE_+\times\RE$. Setting $\rho(q,p) := | f(q,p)|^2$, for all $t\in\RE$ one has that 
\begin{equation}\label{rhot}
\rho_t(q,p) = \left\{ 
\begin{aligned} 
&\rho\Big(q-\frac{pt}m,p\Big) \qquad &\text{if }\,q-\frac{pt}m >0\,, \\ 
& \rho\Big(\!-q+\frac{pt}m,-p\Big) \qquad &\text{if }\,q-\frac{pt}m <0 \,.
\end{aligned}
\right.
\end{equation}
and it  satisfies the equation 
\[
i\pd{}{t}\rho_t = - \,i\,X_{0}\cdot\nabla\rho_t\,,
\]
for all $f\in C^\infty_c(\RE_+\times\RE)$ and $q-\frac{pt}m \neq 0$. 

We remark that, assuming elastic collision at the origin, a classical particle moving on the half-line follows a simple, though possibly discontinuous, trajectory in the phase space: at any time $t\in\RE $ any  initial state $(q,p) \in \RE_+\times \RE$ is mapped into
\begin{equation}\label{varphit}
\varphi_t(q,p) := \left\{ 
\begin{aligned} 
&\Big(q+\frac{pt}m,p\Big) \qquad &\text{if }\,q+\frac{pt}m >0\,, \\ 
& \Big(\!-q-\frac{pt}m,-p\Big) \qquad &\text{if }\,q+\frac{pt}m <0 \,.
\end{aligned}
\right.
\end{equation}
Hence, given a density $\rho : \RE_+\times \RE \to \RE_+$ in the phase space, one has  the identity $\rho_t(q,p) = \rho(\varphi_{-t}(q,p))$ (see  Eqs. \eqref{rhot} and \eqref{varphit}). In this sense,  the group $e^{-i t \LD} $ given in Eq. \eqref{group} describes a classical particle on the half-line. The function $f_t:=e^{-i t \LD}f$ should be interpreted as a classical wave function, with associated probability density function in the phase space given by $\rho_t(q,p) := \big|f_t (q,p)\big|^2$. 

On a graph this interpretation fails when the generator of the dynamics is $e^{-i t L_K}$. In particular, from Eq. \eqref{expLKgroup}, for all $t\in\RE$ it would follow 
\begin{equation*}
\rho_{\ell,t}(q,p) = \big|\big(e^{-i t L_K} \fv\big)_\ell(q,p)\big|^2 =
\left\{ 
\begin{aligned} 
&\bigg|f_\ell\Big(q-\frac{pt}m,p\Big)\bigg|^2 \qquad &\text{if }\,q-\frac{pt}m >0\,, \\ 
&\bigg| \sum_{\ell'}(\MC)_{\ell,\ell'} f_{\ell'}\Big(\!-q+\frac{pt}m,-p\Big)\bigg|^2 \qquad &\text{if }\,q-\frac{pt}m <0 \,,
\end{aligned}
\right.
\end{equation*}
but this ``density'' cannot be understood in terms of a trajectory of the classical particle as $\rho_t(q,p) = \rho(\varphi_{-t}(q,p))$. Also, in general, it   does not coincide with the time evolution of the initial density $|\fv|^2$ as prescribed by Barra and Gaspard, see \cite[Eq. (10)]{BG-pre01_1}. The latter, adapted to our setting and notation, and taking into account the fact that we are considering a non-compact graph,  would give (for all $t \in \RE$)
\begin{equation*}
\rho_{\ell,t}^{BG}(q,p)=
\left\{ 
\begin{aligned} 
& \Big|f_\ell\Big(q-\frac{pt}m,p\Big)\Big|^2 \qquad &\text{if }\,q-\frac{pt}m >0\,, \\ 
& \sum_{\ell'}|(\MC)_{\ell,\ell'}|^2\, \Big|f_{\ell'}\Big(\!-q+\frac{pt}m,-p\Big)\Big|^2 \qquad &\text{if }\,q-\frac{pt}m <0\,.
\end{aligned}
\right.
\end{equation*}
Nevertheless, our definition of the classical dynamics in terms of a classical wave function  turns out to be useful to study the semiclassical limit of the quantum evolution of approximated coherent states. In many ways our approach is similar to the one of Hagedorn \cite{Hag1}. 

In general, a coherent state (on the real-line) is the wave function $\psi^{\hbar}:\RE\to\CO$ defined as 
\[
\psi^{\hbar}(x) = \psi^{\hbar}(\sigma,\sigmap,q,p;x) :=\frac{1}{(2\pi\hbar)^{1/4}\sqrt{\sigma}}\ \exp\left({-\,\frac{\sigmap}{4\hbar \sigma}\,(x-q)^{2}+\frac{i}{\hbar}\,p(x-q)}\right)\,, 
\]
with $p,q\in\RE$ and $\sigma,\sigmap \in\CO\backslash\{0\}$ such that $\Re(\sigmap\sigma^{-1}) = |\sigma|^{-2} >0$. 

In his seminal paper \cite{Hag1},  Hagedorn provides the semiclassical evolution of a coherent state in the presence of  a regular $(C^2(\RE))$ potential. One main observation in \cite{Hag1} is that the quantum evolution of an initial datum of the same form  as above,   is approximated (in $L^2(\RE)$-norm for $\hbar$ small enough) by the wave-function $e^{\frac{i}{\hbar}A_t}\psi^{\hbar}(\sigma_t,\sigmap_t, q_t,p_t;x)$ where: the pair   $(q_t,p_t)$ is determined by the Hamilton equations
\[
\dot p_t=-V'(q_t)\;,\qquad \dot q_t=\frac1{m}\,p_t\,,
\]
here $V$ is the interaction potential and the  initial datum is $(q_0,p_0) = (q,p)$; $\sigma_t$ and $\sigmap_t$ are given by 
\[
\sigma_t = \pd{q_t}{q} \,\sigma + \frac{i}{2} \,\pd{q_t}{p}\, \sigmap\;,\qquad \sigmap_t = -2i\,\pd{p_t}{q}\,\sigma + \pd{p_t}{p}\, \sigmap\,;
\]
and $A_t$ is given by 
\[A_t = \int_0^t ds\,\Big(\frac{p_s^2}{2m} -V(q_s)\Big)\;.\]

In our notation, one can think to a classical function given by $\phi^{\hbar}(\sigma,\sigmap,x;q,p) = \psi^{\hbar}(\sigma,\sigmap, q,p;x)$. With this identification  $\psi^{\hbar}(\sigma_t,\sigmap_t, q_t,p_t;x) = \phi^{\hbar}(\sigma_t,\sigmap_t, x;q_t,p_t) =e^{iL_V t} \phi^{\hbar}(\sigma_t,\sigmap_t,x)(q,p)$, with $L_V$ the Liouville operator associated to the vector field of the classical Hamiltonian $\frac{p^2}{2m} + V(q)$. This is analogous to $ e^{it L_{K}} \bm{\Sigma}^{\hbar}_{\sigma_t,x}(\xi)$. 

\subsection{Coherent states on a star-graph with an even number of edges}

Obviously, by superposition, one could consider initial states of the form 
\begin{equation*}
\bm{\Xi}^{\hbar}_{(\sigma_1,\xi_1),\dots,(\sigma_n,\xi_n)} \equiv 
\begin{pmatrix}\Xi^{\hbar}_{\sigma_1,\xi_1} \\ 
\vdots \\
\Xi^{\hbar}_{\sigma_n,\xi_n}
\end{pmatrix},
\end{equation*}
with $\sigma_\ell>0$ and $\xi_\ell = (q_\ell,p_\ell)\in \RE_+\times \RE$, and results similar to the ones stated in Theorems \ref{t:dynamics} and \ref{t:waveoperators} hold true (with additive error terms). For $n$ even it is possible to  construct states  that propagate exactly like coherent states on the real line.  Suppose $n$ even and  consider a state $ \bm{\mathring{\psi}^\hbar_{\sigma,\xi}}$ defined component-wisely by
\[
 (\bm{\mathring{\psi}^\hbar_{\sigma,\xi}})_{\ell} : = \left\{\! \begin{array}{ll}
\psi^\hbar_{\sigma,\xi} & \qquad \mbox{if }\, \ell = 1,\dots,n/2\,, \\ 
\psi^\hbar_{\sigma,-\xi} & \qquad \mbox{if }\,\ell = n/2 + 1,\dots, n\,.
 \end{array} 
\right. 
\]
It is easy to check that such states belong to the domain of $H_K$. Next, consider the state $ e^{\frac{i}{\hbar}A_t} \bm{\mathring{\psi}^\hbar_{\sigma_t,\xi_t}} $. Taking the time derivative component by component one has 
\[
i\hbar\,\frac{\partial}{\partial t} \, e^{\frac{i}{\hbar}A_t}(\bm{\mathring{\psi}^\hbar_{\sigma_t,\xi_t}})_{\ell} = -\, \frac{\hbar^2}{2m}\,  e^{\frac{i}{\hbar}A_t}(\bm{\mathring{\psi}^\hbar_{\sigma_t,\xi_t}})_{\ell}''\,, 
\]
by the definition of coherent states, see Eq. \eqref{initial}. Since $ e^{\frac{i}{\hbar}A_t}\bm{\mathring{\psi}^\hbar_{\sigma_t,\xi_t}} \in \dom(H_K)$ the latter is equivalent to 
\[
i\hbar\frac{\partial}{\partial t}  e^{\frac{i}{\hbar}A_t}\bm{\mathring{\psi}^\hbar_{\sigma_t,\xi_t}} =H_K\, e^{\frac{i}{\hbar}A_t}\bm{\mathring{\psi}^\hbar_{\sigma_t,\xi_t}}\,. 
\]
Moreover $e^{\frac{i}{\hbar}A_t}\bm{\mathring{\psi}^\hbar_{\sigma_t,\xi_t}} \big|_{t=0} =  \bm{\mathring{\psi}^\hbar_{\sigma_0,\xi}} $. Hence, $e^{-i\frac{t}{\hbar}H_{K}} \bm{\mathring{\psi}^\hbar_{\sigma_0,\xi}} = e^{\frac{i}{\hbar}A_t} \bm{\mathring{\psi}^\hbar_{\sigma_t,\xi_t}} $.  On the other hand, define the classical state $\bm{\mathring{\phi}^\hbar_{\sigma,x}}$, component-wisely, by $(\bm{\mathring{\phi}^\hbar_{\sigma,x}})_\ell(\xi) := (\bm{\mathring{\psi}^\hbar_{\sigma,\xi}})_\ell(x)$. Noting the identity $\sum_{\ell'=1}^n\big(\frac{2}{n} - \delta_{\ell,\ell'}\big) (\bm{\mathring{\phi}^\hbar_{\sigma,x}})_{\ell'}(- \xi) = (\bm{\mathring{\phi}^\hbar_{\sigma,x}})_\ell(\xi)$, and by using Eq. \eqref{expLKgroup}, it is easy to infer the identity $\big(e^{itL_{K}} \bm{\mathring{\phi}^\hbar_{\sigma_t,(\cdot)}}\big)(\xi) = \bm{\mathring{\phi}^\hbar_{\sigma_t,(\cdot)}}(\xi_t) = \bm{\mathring{\psi}^\hbar_{\sigma_t,\xi_t}}$, so that 
\[
e^{-i\frac{t}{\hbar}H_{K}}\,\bm{\mathring{\psi}^\hbar_{\sigma_0,\xi}}  = e^{\frac{i}{\hbar}\Szero_{t}} \big(e^{itL_{K}} \bm{\mathring{\phi}^\hbar_{\sigma_t,(\cdot)}}\big)(\xi)\,, 
\]
which is equivalent to \eqref{exact}. In this sense, up to a normalization factor $\sqrt{2/n}$, states of the form $\mathring{\bm{\psi}}^\hbar_{\sigma,\xi}$ are coherent states on a star-graph with an even number of edges.

\appendix
\section{Wave operators for the pair (Dirichlet Laplacian, Neumann Laplacian) on the half-line\label{app}}

\begin{proposition}\label{p:WND} Let  $\Omega^{\pm}_{ND}$ be the wave operators for the pair (Dirichlet Laplacian, Neumann Laplacian) in $L^2(\RE_+)$, defined by 
\[
\Omega^{\pm}_{ND}:=\slim_{t \to \pm \infty} \UN_{-t}\, \UD_t \,.
\]
There holds true: 
\[
\Omega^{\pm}_{ND} =  \pm \mathcal F_c^* \mathcal F_s \;. 
\]
\end{proposition}
\begin{proof}
We use a density argument. Let $\psi\in L^2(\RE_+)$.  For any $\eps>0$ there exists $\varphi \in C_0^\infty(\RE_+) $ such that $\|\psi - \varphi\|_{L^2(\RE_+)} \leq \eps/4$. Recalling the trivial bounds $\|U^{D/N}_{t}\| \leq 1$ and $\|\mathcal F_{s/c} \| \leq 1$, we infer 
\[
\|(\UN_{-t}\,\UD_t \mp \mathcal F_c^* \mathcal F_s)\psi\|_{L^2(\RE_+)} \leq \frac{\eps}{2} +  \|(\UN_{-t}\, \UD_t \mp \mathcal F_c^* \mathcal F_s)\varphi\|_{L^2(\RE_+)}\,. 
\]

Hence, it is enough to prove that 
\begin{equation}\label{gilmour}
\lim_{t\to \pm\infty } \|(\UN_{-t}\, \UD_t \mp \mathcal F_c^* \mathcal F_s)\varphi\|_{L^2(\RE_+)} = 0 
\qquad \forall \varphi \in C_0^\infty(\RE_+) \,.
\end{equation}

Note that
\[
\|(\UN_{-t}\, \UD_t \mp \mathcal F_c^* \mathcal F_s)\varphi\|_{L^2(\RE_+)} = \|( \UD_t \mp\UN_{t} \mathcal F_c^* \mathcal F_s)\varphi\|_{L^2(\RE_+)}  = \|( \UD_t \mathcal F_s^*\mp\UN_{t} \mathcal F_c^* )\mathcal F_s\varphi\|_{L^2(\RE_+)} \,.
\]

Moreover, for all $t \in \RE$ the following identities hold true:
\begin{gather*}
(\UD_t \psi)(x) = \frac{2i}{\sqrt{2\pi}} \int_0^\infty\!\! dk \; e^{-ik^2 t }\sin (kx)\, (\mathcal F_s \psi)(k)\,;\\
(\UN_t \psi)(x) = \frac{2}{\sqrt{2\pi}} \int_0^\infty\!\! dk \; e^{-ik^2 t }\cos (kx)\, (\mathcal F_c \psi)(k)\,.
\end{gather*}
Hence,
\begin{gather*}
(\UD_t \mathcal F_s^*\psi)(x) = \frac{2i}{\sqrt{2\pi}} \int_0^\infty\!\! dk \; e^{-ik^2 t }\sin (kx)\, \psi(k)\,,\\
(\UN_t \mathcal F_c^*\psi)(x) = \frac{2}{\sqrt{2\pi}} \int_0^\infty\!\! dk \; e^{-ik^2 t }\cos (kx)\, \psi(k)\, ,
\end{gather*}
which give
\[\begin{aligned}
\big(( \UD_t \mathcal F_s^*\mp\UN_{t} \mathcal F_c^* )\mathcal F_s\varphi \big) (x) = &   \frac{2}{\sqrt{2\pi}} \int_0^\infty\!\! dk\; (i\sin (kx)  \mp \cos(kx))\, e^{-ik^2 t }\,(\mathcal F_s \varphi)(k)  \\ 
= &   \mp \frac{2}{\sqrt{2\pi}} \int_0^\infty\!\! dk\; e^{\mp ikx-ik^2 t }\,(\mathcal F_s \varphi)(k) \,. 
\end{aligned}\]

We have obtained the following explicit formula for the quantity we are interested in 
\[
\|(\UN_{-t} \UD_t \mp \mathcal F_c^* \mathcal F_s)\varphi\|_{L^2(\RE_+)}^2 =
\frac{2}{\pi}\int_0^\infty\!\! dx \left|\int_0^\infty\!\! dk\;   e^{\mp ikx-ik^2 t }\,(\mathcal F_s \varphi)(k) \right|^2 .
\]

We note that, to prove the statement  for $W_{+}^{ND}$ we have to study  the limit $t\to +\infty$ of  
\begin{equation}\label{david}
\|(\UN_{-t} \UD_t -  \mathcal F_c^* \mathcal F_s)\varphi\|_{L^2(\RE_+)}^2 =
\frac{2}{\pi}\int_0^\infty\!\! dx \left|\int_0^\infty\!\! dk\;  e^{-  ikx-ik^2 t }\,(\mathcal F_s \varphi)(k) \right|^2 ,
\end{equation}
while, to prove the statement for $W_{-}^{ND}$ we have to study  the limit $t\to -\infty$ of  
\[
\|(\UN_{-t} \UD_t + \mathcal F_c^* \mathcal F_s)\varphi\|_{L^2(\RE_+)}^2 =
\frac{2}{\pi}\int_0^\infty\!\! dx \left|\int_0^\infty\!\! dk\;   e^{ ikx+ik^2 |t| }\,(\mathcal F_s \varphi)(k) \right|^2 . 
\]

In what follows we focus the attention on the limit $t\to+\infty$. The other limit is obtained with trivial modifications. Hence, from now on we assume $t>0$. From Eq. \eqref{david},  changing variables  $k\to \eta  = k\,t^{1/2}$ and $x\to y = x/t^{1/2}$ , we obtain 
\[\begin{aligned}
\|(\UN_{-t} \UD_t - \mathcal F_c^* \mathcal F_s)\varphi\|_{L^2(\RE_+)}^2 
& = \frac{2}{\pi\,t^{1/2}} \int_0^\infty\!\! dy \left|\int_0^\infty\!\! d\eta \;   e^{- i\eta y -i\eta^2 }\,(\mathcal F_s \varphi)(\eta /t^{1/2}) 
 \right|^2 \\
& = \frac{2}{\pi\,t^{1/2}} \int_0^1\! dy \left|F(y, t)\right|^2  + \frac{2}{\pi} \frac1{t^{1/2}}\int_1^\infty\!\! dy \left|F(y, t)\right|^2 ,
\end{aligned}\]
with  
\[
F(y, t): =\int_0^\infty\!\! d\eta \;   e^{-i\eta y -i\eta^2 }\,(\mathcal F_s \varphi)(\eta /t^{1/2}) \;. 
\]

For  any $\varphi \in C_0^\infty(\RE_+) $, $\mathcal F_s \varphi$ belongs to $C^\infty(\RE_+)$, it decays at infinity  faster than any polynomial  in  $1/k$,   $(\mathcal F_s \varphi)(0) = 0$, moreover 
\[ 
|(\mathcal F_s \varphi)(k)| \leq  \frac{2 k}{\sqrt{2\pi}} \int_0^\infty\!\! dx\; x\, |\varphi (x)| \,.
\]
Hence, 
\begin{equation}\label{Cvarphi1}
\int_0^\infty\!\! dk\; \frac1{k^\delta}\, |\mathcal F_s \varphi (k) |< \infty \qquad \text{for all $\delta < 2 $}\,. 
\end{equation}
 
Additionally,  for  any $\varphi \in C_0^\infty(\RE_+) $, $\mathcal F_c \varphi$ belongs to $C^\infty(\RE_+)$, it   decays at infinity  faster than any polynomial  in  $1/k$, and  
$|(\mathcal F_c \varphi)(k)| \leq  \frac{2 }{\sqrt{2\pi}} \int_0^\infty\! dx\, |\varphi (x)|$.
Hence, 
\begin{equation}\label{Cvarphi2}
\int_0^\infty\!\! dk\; \frac1{k^\delta}\, |\mathcal F_c \varphi(k)| < \infty \qquad \text{for all $\delta <1 $}\,. 
\end{equation}

Starting from the identity $e^{-i\eta y -i\eta^2 } = \frac{i}{y+2\eta}\frac{d}{d\eta}e^{-i\eta y -i\eta^2 }$, by integration by parts, we obtain 
\[
F(y, t) =F_1(y, t) + F_2(y,t) \,,
\]
with 
\[
F_1(y, t) : = \int_0^\infty\!\! d\eta \;   e^{-i\eta y -i\eta^2 } \frac{2i}{(y+2\eta)^2}\,(\mathcal F_s \varphi)(\eta /t^{1/2})
\] 
and 
\[ 
F_2(y, t) : =\frac1{t^{1/2}} \int_0^\infty\!\! d\eta \;   e^{-i\eta y -i\eta^2 } \frac{1}{i(y+2\eta)}\,\big(\mathcal F_c((\cdot) \varphi)\big)(\eta /t^{1/2})\,.
\]

Since $y$ and $\eta$ are both positive, from the trivial inequality $\frac{1}{(y+2\eta)^a} \leq \frac{1}{y^b} \frac{1}{(2\eta)^c}$  for all $y, \eta>0$ and for all $a,b,c>0$ such that  $a = b+c$, we deduce that 
\[
|F_1(y, t)| \leq \frac{C}{y^{1/4}} \int_0^\infty\!\!  d\eta\; \frac{1}{\eta^{7/4}}\, |(\mathcal F_s \varphi)(\eta /t^{1/2})| =  \frac{C}{t^{3/8}\,y^{1/4}} \int_0^\infty\!\!  d\xi\;\frac{1}{\xi^{7/4}}\, |(\mathcal F_s \varphi)(\xi )|  =  \frac{C}{t^{3/8}\,y^{1/4}}\,;
\]
here and in the following $C$ denotes a generic positive constant whose value may depend only on integrals of the sine (or cosine) Fourier transform of $\varphi$ (or $(\cdot)\varphi$), as in Eqs. \eqref{Cvarphi1} and \eqref{Cvarphi2}; the value of $C$ may change from line to line and precise values for the constants can be obtained. In a similar way we infer, 
\[ 
|F_2(y, t)| \leq \frac{C}{t^{1/2}\, y^{1/4}} \int_0^\infty\!\! d\eta \;    \frac{1}{\eta^{3/4}}\,\big|\big(\mathcal F_c((\cdot) \varphi)\big)(\eta /t^{1/2})\big| = 
\frac{C}{t^{3/8}\, y^{1/4}} \int_0^\infty\!\! d\xi \;    \frac{1}{\xi^{3/4}}\,\big|\big(\mathcal F_c((\cdot) \varphi)\big)(\xi )\big| = \frac{C}{t^{3/8}\,y^{1/4}} \,. 
\]

Hence, 
\[
\frac{2}{\pi\,t^{1/2}}\int_0^1\! dy\, \left|F(y, t)\right|^2\leq \frac{C}{t^{5/4}} \int_0^1\! dy\; \frac{1}{y^{1/2}} \leq \frac{C}{t^{5/4}}\, .
\]

On the other hand, for $y>1$, 
\[
|F_1(y, t)| \leq C\frac{1}{y} \int_0^\infty\!\!  d\eta\; \frac{1}{\eta} \,|(\mathcal F_s \varphi)(\eta /t^{1/2})| = \frac{C}{y} \int_0^\infty\!\!  d\xi\;\frac{1}{\xi}\, |(\mathcal F_s \varphi)(\xi )|  =  \frac{C}{y} 
\]
and 
\[ 
|F_2(y, t)| \leq \frac{C}{t^{1/2}\,y} \int_0^\infty\!\! d\eta \; \big|\big(\mathcal F_c((\cdot) \varphi)\big)(\eta /t^{1/2})\big| = 
\frac{C}{y}  \int_0^\infty\!\! d\xi\; \big|\big(\mathcal F_c((\cdot) \varphi)\big)(\xi )\big| = \frac{C}{y} \,.
\]
So, we obtain
\[
\frac{2}{\pi\,t^{1/2}}\int_1^\infty\!\! dy\, \left|F(y, t)\right|^2\leq \frac{C}{t^{1/2}} \int_1^\infty\!\!  dy\; \frac{1}{y^2 }\leq \frac{C}{t^{1/2}} \,.
\]

In this way we have proved that, for all $\varphi \in C_0^\infty(\RE_+)$  there exists a constant $C$ such that
\[
\|(\UN_{-t}\, \UD_t - \mathcal F_c^* \mathcal F_s)\varphi\|_{L^2(\RE_+)} \leq \frac{C}{t^{1/4}} \qquad \mbox{for all $t>1$}\,,
\]
and the latter claim  gives the limit in Eq. \eqref{gilmour} for $t\to+\infty$. 
\end{proof}


\begin{thebibliography}{99}
\bibitem{ACFN11} R. Adami, C. Cacciapuoti, D. Finco, D. Noja: Fast solitons on star graphs. \textit{Rev. Math. Phys.} \textbf{23}(04) (2011), 409--451.

\bibitem{BG-pre01_2} F. Barra, P. Gaspard: Transport and dynamics on open quantum graphs. \textit{Phys. Rev. E} \textbf{65} (2001),  016205 (21 pages).

\bibitem{BG-pre01_1} F. Barra, P. Gaspard: Classical dynamics on graphs. \textit{Phys. Rev. E} \textbf{63} (2001),  066215 (22 pages).

\bibitem{BK}  G. Berkolaiko, P. Kuchment: {\it Introduction to quantum graphs.} Mathematical Surveys and Monographs Vol. {\bf 186}, American Mathematical Society, Providence, RI (2013). 

\bibitem{CFP19} C. Cacciapuoti, D. Fermi, A. Posilicano: The semi-classical limit with delta potentials, arXiv:1907.05801 [math-ph] (2019), 31 pages.

\bibitem{C-psim10} V. L. Chernyshev: Time-dependent Schr\"odinger equation: Statistics of the distribution of Gaussian packets on a metric graph. (Russian) Tr. Mat. Inst. Steklova {\bf 270}  (2010) 249--265. Translation in Proc. Steklov Inst. Math. {\bf 270}(1) (2010), 246--262.

\bibitem{CS-mz07} V. L. Chernyshev, A. I. Shafarevich: The semiclassical spectrum of the Schr\"odinger operator on a geometric graph. (Russian) Mat. Zametki {\bf 82}(4) (2007), 606--620. Translation in Math. Notes {\bf 82}(3-4) (2007), 542--554.

\bibitem{CS-rjmp08} V. L. Chernyshev, A. I. Shafarevich:   Semiclassical asymptotics and statistical properties of Gaussian packets for the nonstationary Schr\"odinger equation on a geometric graph. Russ. J. Math. Phys. {\bf 15}(1) (2008),  25--34.

\bibitem{CS-ptrs14}  V. L. Chernyshev, A. I. Shafarevich: Statistics of Gaussian packets on metric and decorated graphs. Philos. Trans. R. Soc. Lond. Ser. A Math. Phys. Eng. Sci. {\bf 372}(2007)  (2014), 20130145 (11 pages).

\bibitem{EN} K.-J. Engel, R. Nagel: {\it One-Parameter Semigroups for Linear Evolution Equations.} Springer, Berlin (2000).

\bibitem{GOO} B. Gaveau, M. Okada, T. Okada: Explicit heat kernels on graphs and spectral analysis, in J. E. Fornaess (Ed.), {\it Several complex variables} (Proceedings of the Mittag-Leffler Institute, Stockholm, 1987-1988), Princeton Math. Notes Vol. {\bf 38}, Princeton University Press (1993).

\bibitem{Hag1} G. A. Hagedorn: Semiclassical Quantum Mechanics I. The $\hbar \to 0$ Limit for Coherent States. {\it Comm. Math. Phys.} {\bf 71} (1980), 77--93.

\bibitem{HS2012} P. R. Halmos, V. S. Sunder: {\it Bounded Integral Operators on $L^2$ Spaces}. Springer-Verlag (1978).

\bibitem{GS-ap06} S. Gnutzmann, U. Smilansky: Quantum graphs: Applications to quantum chaos and universal spectral statistics. {\it Adv. Phys.} {\bf 55} (2006), 527--625.

\bibitem{KosSch} V. Kostrykin, R. Schrader: Laplacians on metric graphs: eigenvalues, resolvents and semigroups, pp. 201--226 in G. Berkolaiko, R. Carlson, S. Fulling and P. Kuchment (Eds.), {\it Quantum Graphs and Their Applications}. Contemporary Math. Vol. {\bf 415}, American Math. Society, Providence, RI (2006).

\bibitem{KS-jpa99}  V. Kostrykin, R. Schrader: Kirchhoff's rule for quantum wires. {\it J. Phys. A} {\bf 32}(4)  (1999), 595--630. 

\bibitem{KS-prl97} T. Kottos, U. Smilansky: Quantum Chaos on Graphs. {\it Phys. Rev. Lett.} {\bf 79}(24) (1997), 4794--4797.

\bibitem{KS-prl00} T. Kottos, U. Smilansky: Chaotic Scattering on Graphs. {\it Phys. Rev. Lett.} {\bf 85}(5) (2000), 968--971.

\bibitem{Maslov} V. P. Maslov: {\it The complex WKB method for nonlinear equations I. Linear theory}.  Birkh\"auser Verlag, Basel
(1994).

\bibitem{P01} A. Posilicano: A Kre\u\i n-like formula for singular perturbations of self-adjoint operators and applications. {\it J. Funct. Anal.} {\bf 183} (2001),  109--147.

\bibitem{P08} A. Posilicano: Self-adjoint extensions of restrictions. {\it Oper. Matrices} {\bf 2}(4) (2008), 483--506.

\bibitem{Weder15} R. Weder: Scattering theory for the matrix Schrödinger operator on the half line with general boundary conditions. {\it  J. Math. Phys.} {\bf56}(9) (2015), 092103 (24 pages). Erratum {\it J. Math. Phys.} {\bf 60}(1) (2019), 019901 (1 page). 
\end{thebibliography}
\end{document}